\documentclass[format=acmsmall]{acmart}
\usepackage{acm-ec-26_camera}
\usepackage{booktabs} 
\usepackage{natbib}
\setcitestyle{authoryear}

\usepackage[multiple]{footmisc} 

\usepackage{tikz}
\usetikzlibrary{arrows.meta}
\usepackage{subcaption}
\usepackage{geometry}

\usepackage[multiple]{footmisc} 
\usepackage[cm,headings]{fullpage}
\usepackage{amsfonts, amsthm, amsmath}
\usepackage{comment, verbatim}
\usepackage{enumitem}
\DeclareMathOperator{\argmax}{arg\,max}

\usepackage{algorithm}
\usepackage{algpseudocode}
\makeatletter
\renewcommand{\ALG@name}{Protocol}
\makeatother
\algtext*{EndFunction} 
\algtext*{EndProcedure} 
\algtext*{EndFor} 
\renewcommand{\textproc}{\textsf}

\usepackage{xcolor}
\usepackage{color}
\usepackage[colorinlistoftodos,textsize=tiny]{todonotes}
\newcommand{\Comments}{0}
\definecolor{gray}{gray}{0.5}
\definecolor{lightred}{rgb}{1,0.6,0.6}
\definecolor{darkgreen}{rgb}{0,0.5,0}
\definecolor{myorange}{rgb}{0.8,0.7,0.5}
\definecolor{darkblue}{rgb}{0.0,0.0,0.5}

\newcommand{\mynote}[2]{\ifnum\Comments=1\textcolor{#1}{#2}\fi}
\newcommand{\mytodo}[2]{\ifnum\Comments=1
  \todo[linecolor=#1!80!black,backgroundcolor=#1,bordercolor=#1!80!black]{#2}\fi}

\ifnum\Comments=1               
\fi

\usepackage{hyperref}
\hypersetup{colorlinks=true,breaklinks=true,bookmarks=true,urlcolor=darkblue,citecolor=darkblue,linkcolor=darkblue,bookmarksopen=false,draft=false}

\newtheorem{theorem}{Theorem}
\newtheorem*{remark*}{Remark}

\newtheorem{lemma}{Lemma}

\newtheorem{proposition}{Proposition}

\newtheorem*{proposition*}{Proposition}

\newtheorem{definition}{Definition}
\newtheorem{axiom}{Axiom}
\usepackage{amsfonts}
\let\oldtextsc\textsc
\renewcommand{\textsc}[1]{\text{\oldtextsc{\upshape #1}}}
\usepackage{dsfont}
\usepackage{bm}
\newcommand{\ones}{\vec{1}}
\newcommand{\indvec}[1]{\boldsymbol{\delta}^{#1}}

\newcommand{\reals}{\mathbb{R}}

\let\C\undefined                
\newcommand{\C}{\mathcal{C}}
\newcommand{\G}{\mathcal{G}}

\newcommand{\N}{\mathbb{N}}

\newcommand{\Y}{\mathcal{Y}}

\newcommand{\Ginit}{\mathcal{G}_\mathrm{init}}
\newcommand{\GLP}{\mathcal{G}_\mathrm{LP}}
\newcommand{\Cinit}{\mathcal{C}_\mathrm{init}}
\newcommand{\CLP}{\mathcal{C}_\mathrm{LP}}

\newcommand{\Span}{\mathrm{span}\,}

\newcommand{\conv}{\mathrm{conv}\,}

\newcommand{\dom}{\mathrm{dom}\,}
\newcommand{\relint}{\mathrm{relint}\,}
\renewcommand{\bar}[1]{\overline{#1}}

\newcommand{\E}{\mathbb{E}}

\newcommand{\fee}{\mathsf{fee}}

\newcommand{\liabilityf}{\mathsf{liability}}
\newcommand{\pricef}{\mathsf{price}}

\usepackage{xspace}

\renewcommand{\vec}[1]{{\mathbf{#1}}}
\renewcommand{\r}{\vec{r}}
\newcommand{\p}{\vec{p}}
\newcommand{\q}{\vec{q}}

\renewcommand{\v}{\vec{v}}
\newcommand{\x}{\vec{x}}
\newcommand{\y}{\vec{y}}
\newcommand{\0}{\vec{0}}
\newcommand{\dG}{\vec{d} G}
\newcommand{\res}{\x}
\newcommand{\xa}{x_1}
\newcommand{\xb}{x_2}
\newcommand{\inprod}[2]{\left\langle #1, #2 \right\rangle}

\makeatletter
\newcommand{\ooverline}[1]{\overline{\dbl@overline{#1}}}
\newcommand{\dbl@overline}[1]{\mathpalette\dbl@@overline{#1}}
\newcommand{\dbl@@overline}[2]{
  \begingroup
  \sbox\z@{$\m@th#1\overline{#2}$}
  \ht\z@=\dimexpr\ht\z@-2\dbl@adjust{#1}\relax
  \box\z@
  \ifx#1\scriptstyle\kern-\scriptspace\else
  \ifx#1\scriptscriptstyle\kern-\scriptspace\fi\fi
  \endgroup
}
\newcommand{\dbl@adjust}[1]{
  \fontdimen8
  \ifx#1\displaystyle\textfont\else
  \ifx#1\textstyle\textfont\else
  \ifx#1\scriptstyle\scriptfont\else
  \scriptscriptfont\fi\fi\fi 3
}
\makeatother

\title{A General Theory of Liquidity Provisioning\\for Prediction Markets}
\author{Adithya Bhaskara}\affiliation{\institution{University of Colorado Boulder}\city{Boulder}\state{Colorado}\country{USA}}
\author{Rafael Frongillo}\affiliation{\institution{University of Colorado Boulder}\city{Boulder}\state{Colorado}\country{USA}}
\author{Elias Lindgren}\affiliation{\institution{University of Colorado Boulder}\city{Boulder}\state{Colorado}\country{USA}}
\author{Maneesha Papireddygari}\affiliation{\institution{Boston College}\city{Chestnut Hill}\state{Massachusetts}\country{USA}}

\begin{abstract}
  Liquidity provisioning in automated market makers is the practice of recruiting third-party liquidity providers (LPs) to contribute assets to the market in exchange for fees skimmed off of trades.
  This paper introduces a general framework for liquidity provisioning in cost function prediction markets.
  Our most general protocol allows LPs to submit or update an arbitrary cost function that specifies their liquidity over the entire price space.
  We show that our protocol encapsulates several notions of running market makers in parallel, which we prove to be equivalent.
  We also recover existing protocols from decentralized finance as special cases.
  In our protocol, liquidity can be expressed as a matrix-valued function, which we argue is necessary with three or more securities.
  Due to this inherent multidimensionality, the design of trading fees with three or more securities is nontrivial: we show that natural axioms on the design of these fees are incompatible.
\end{abstract}
\begin{document}

\maketitle

\subsection*{Acknowledgments}

This material is based upon work supported by the Ethereum Foundation and the National Science Foundation under Grant Nos.\ IIS-2045347 and DMS-1928930, the latter while the second author was in residence at the Mathematical Sciences Research Institute in Berkeley, California, during the Fall 2023 semester.
Part of this research was conducted when the fourth author was an intern at the Ethereum Foundation.
We thank Gabriel Andrade, Davide Crapis, James Grugett, Ciamac Moallemi, Alex Solleiro, and Bo Waggoner for several interesting discussions and ideas.
We also thank Guillermo Angeris for suggestions regarding the piecewise linear market maker and anonymous reviewers for their feedback.

\newpage
\section{Introduction}

\citet{hanson2003combinatorial} introduced the concept of an automated market maker to solve thin market problems in combinatorial prediction markets.
In contrast to order books and continuous double auctions, where buyers are matched to sellers, automated market makers are central authorities willing to price any bundle of assets to buy or sell.
More recently, automated market makers have gained in popularity in the context of decentralized finance as a low-gas way to implement a market on a blockchain \citep{base2025gas,bartoletti2022ammdefi,xu2023dexamms,mohan2022ammdefi}.
Along with this trend came the innovation of decoupling the roles of the market mechanism, which facilitates trade, and liquidity providers (LPs), which take on risk to stabilize prices.
In this new paradigm, the market mechanism exposes another interface to potential LPs, which may deposit assets in exchange for a cut of the fees charged on the trades using those assets.

Thus far, however, the design of liquidity provisioning interfaces has been somewhat \emph{ad hoc}, limited to its decentralized finance origins, and focused only on the case of two assets.
For example, in the Uniswap V2 interface, LPs must deposit funds proportional to the current reserves, a natural yet restrictive interface.
\citet{zetlinjones2024automated} show how these restrictions lead LPs to actively trade against the market---that is, themselves and other LPs---leading to potential inefficiencies.
Uniswap V3 adds significant flexibility, but the interface is still somewhat cumbersome:
LPs must contend with discrete ``buckets'' in the price space to allocate their funds.
Throughout, the full design space of liquidity provisioning protocols has been far from clear.

This gap is especially large in the case of prediction markets, which often exchange more than two securities.
For example, an election market might offer a security for each potential candidate, or even a combinatorial market for the outcomes per state.
Offering multiple independent markets for each pair of securities not only leads to information loss, but also creates large arbitrage opportunities that increase the risk of providing liquidity~\citep{dudik2012tractable}.
Thus, effective liquidity provisioning protocols for more than two assets could have a significant impact on the performance and prevalence of prediction markets.
Despite all these considerations, no liquidity provisioning protocol for prediction markets has been proposed.
\subsection{Contributions}

To fill this gap, we introduce a general framework for liquidity provisioning protocols for cost function prediction markets trading any number of securities (\S~\ref{sec:protocol-pred-markets}).
Our protocol allows LPs to submit an arbitrary cost function, specifying their liquidity over the entire price space.
Then, the mechanism determines the deposit required to ensure that LPs never owe the market maker in any future state.
Our framework is built on the idea of LPs running market makers ``in parallel,'' and we show that several notions of parallel market making are equivalent, including the scoring rule markets of \citet{hanson2003combinatorial} (\S~\ref{sec:equivalence}).

We also address various difficulties that arise in the case of three or more securities.
First, we justify the need for our general protocol by demonstrating the inherent multidimensionality of liquidity (\S~\ref{sec:liquidity_discussion}).
The liquidity at a given price must be described by a full matrix, allowing one to assess the liquidity in each ``direction,'' i.e., for each possible trade.
As a corollary, any fully general liquidity provisioning protocol must allow liquidity between all assets simultaneously, rather than only allowing liquidity in pairwise markets.
Moreover, we show that natural axioms on the design of trading fees are incompatible when there are three or more securities traded (\S~\ref{sec:fees}).

Finally, we additionally show that in the two-asset case, our framework recovers several existing protocols (\S~\ref{sec:DeFi}) used in decentralized finance.
We also provide several restricted protocols that are computationally feasible.
Furthermore, we contribute to the decentralized finance literature by giving a fully expressive protocol for the trade of any number of assets.

\subsection{Related work}

We direct readers to \citet{chen2007utility,abernethy2013efficient} for an overview of the literature on automated market makers for prediction markets, and to \citet{angeris2020improved} for those in decentralized finance.
Our work heavily relies on \citet{frongillo2023axiomatic}, which establishes the equivalence of (non-parallelized) automated market makers as used for prediction markets to those used in decentralized finance.
One can view our analysis of the trade split $\r=\sum_i\r^i$ as a special case of optimal routing problems stated in~\cite{angeris2022optimal,diamandis2023efficient}. 
We provide a closed form solution to this special case not yet seen in the literature.
This work is also related to recent explorations of running parallel LMSRs \cite{dudik2021logtime} and of geometric aspects of automated market makers~\citep{angeris2023geometry}. 
In particular, the Minkowski sum operations in the latter paper can be seen as implicitly computing an infimal convolution, which they also view as a combined market maker.
However, their connections to implementing liquidity provisioning are not explored beyond a very restricted setting. 
Additionally, \citet{ramseyer2024augmenting} study similar notions of price coherence and parallel market makers to ours, but while working to design batch exchanges that support constant function market makers for liquidity provision.

One could view our general framework as theoretically formalizing Minswap~\citep{nguyen2021minswap}, which is designed on Cardano to accommodate multiple LP pools.
Perhaps closest to our work is \citet{milionis2023complexity}, which asks LPs for their ``demand curves'' to be aggregated into a two-asset market maker.
Our framework can be thought of as a more general way of thinking about LPs running several markets in parallel.
Their demand curves $h(p)$ (denoted $g(p)$ in their paper) are related to our market scoring rule with $h(p)=S(p,1)$ and $-\int pdh(p)=S(p,0)$.
While demand curves are well-rooted in microeconomic foundations \citep{milionis2024myersonian}, we demonstrate in \S~\ref{sec:liquidity_discussion} that cost functions (or some other suitable higher-dimensional notion of demand curves) are necessary to give a fully general liquidity provisioning protocol for more than two assets, a crucially important case for prediction markets as we describe above.
They are also instrumental in helping us propose a closed-form construction of aggregate CFMMs and in helping us realize and prove different equivalent interpretations (Theorem~\ref{thm:equivalence-interpretations}).
We additionally handle the technical issues that arise with more than two securities (Theorem~\ref{thm:fee-impossible}).

There has additionally been some work on the \emph{incentives} and \emph{effects} that trading fees have on traders and LPs' behavior~\citep{angeris2024multidimensional,campbell2025optimal,milionis2025automated}, but the design of trading fees themselves has been unexplored.
This gap is especially large for three or more securities, where we show several natural properties to be incompatible.

\section{Background}

\subsection{Notation and convex analysis}
Vectors are denoted in bold, e.g., $\q \in \reals^n$, and $q_i\in \reals$ denotes the $i$th coordinate of $\q$.
The all-zeros vector is $\0 = (0,\dots,0)$ and the all-ones vector is $\ones = (1,\dots,1)$. We define the indicator vector $\indvec{i}$ by $\delta^i_i = 1$ and $\delta^i_j = 0$ for $j\neq i$.
Comparison between two vectors is pointwise, e.g., $\q\succ\hat{\q}$ if $q_i>q_i'$ for all $i=1,\ldots,n$, and similarly for $\succeq$.
We say $\q\gneqq\hat{\q}$ when $q_i\geq q_i'$ for all $i$ and $\q\neq\hat{\q}$. Define $\reals_{\geq0}^n=\{\q\in\reals^n\mid\q\succeq\0\}$, $\reals_{>0}^n=\{\q\in\reals^n\mid\q\succ\0\}$, \emph{et cetera}.
Finally, we denote the probability simplex by $\Delta_n = \{\p \in \reals^n_{\geq 0} \mid \inprod{\p}{\ones} = 1\}$.
Let $f:\reals^n\to\reals$. We use the following conditions:
\begin{itemize}[itemsep=0pt]
\item \textit{convex}: $\forall \x,\y\in\reals^n,\lambda\in[0,1],f(\lambda \x+(1-\lambda)\y)\leq\lambda f(\x)+(1-\lambda)f(\y)$.
\item \textit{$\ones$-invariant}: $f(\q+\alpha\ones)=f(\q)+\alpha$ for all $\q\in\reals^n$, $\alpha\in\reals$.
\item \textit{1-homogeneous (on $\reals_{\geq 0}^n$)}: $f(\alpha\q)=\alpha f(\q)$ for all $\q\succeq\0$ and $\alpha>0$.
\end{itemize}
We write $f'$ to indicate differentiation when $f$ is 1-dimensional.

\begin{definition}[Convex conjugate]
  For a function $f:\reals^n\to\reals\cup\{\infty\}$ we define its \emph{convex conjugate} $f^*:\reals^n\to\reals\cup\{\infty,-\infty\}$ by $f^*(\x^*) = \sup_{\x\in\reals^n} \inprod{\x^*}{\x} - f(\x)$.
\end{definition}

\begin{definition}[Subgradient]
  For a function $f:\reals^n\to\reals\cup\{\infty\}$ and $\x\in\reals^n$ we define the set of \emph{subgradients} of $f$ at $\x$ by $\partial f(\x) = \{\x^* \in \reals^n \mid \forall \x'\in\reals^n\, f(\x') - f(\x) \geq \inprod{\x^*}{\x'-\x}\}$.
\end{definition}

\begin{definition}[Infimal convolution]
  For functions $f_i:\reals^n\to\reals\cup\{\infty\}$ we define their \emph{infimal convolution} $f = \bigwedge_i f_i$ by $f(\x) = \inf\left\{\sum_i f_i(\x^i) \mid \sum_i \x^i = \x\right\}$, where the $\x^i$ range over $\reals^n$.
\end{definition}

\begin{definition}[1-homogeneous extension $\overline f$]
  Given $f:\Delta_n\to\reals$, we define its \emph{1-homogeneous extension} $\overline f : \reals^n_{\geq 0} \to \reals$ by $\overline f(\x) := \|\x\|_1 f(\x/\|\x\|_1)$ for $\x \neq \0$ and $\overline f(\0) = 0$.
\end{definition}

\subsection{Cost functions and prediction markets} \label{sec:predmarketscostfuncs}

Automated market makers (AMMs) are mechanisms that are always willing to trade a bundle of $n$ securities for some price.
In contrast to traditional order book settings, where buyers and sellers must be matched, traders can trade with AMMs directly.
Prediction markets are AMM mechanisms that seek to elicit probability distributions over future events by allowing traders to buy and sell securities.
Some common instantiations of AMM prediction markets can be seen in horse betting, Iowa electronic markets, and more recently, \citet{manifold2022maniswap} and \citet{andrade2025lmsr}.

Suppose a random variable $Y$ about a future event takes values from set $\Y$, which contains $n$ mutually exclusive and exhaustive outcomes.
A trader holding an Arrow-Debreu (AD) security associated with $y\in\Y$ gets paid \$1 when outcome $y$ happens and \$0 otherwise. 
A security market is \textit{complete} if it trades $n$ independent AD securities, one for each outcome.
While our Protocol~\ref{alg:general-n-asset-cost} works for incomplete markets, we restrict our attention to complete markets in this paper for ease of exposition.
Our proofs in \S~\ref{sec:equiv-interpr} rely heavily on the market being complete.
But, in Appendix~\ref{sec:app-incomplete-markets}, we discuss how our framework extends to the incomplete case.

\citet{chen2007utility,chen2013cost,abernethy2013efficient} characterize prediction markets, and \citet{abernethy2013efficient} show that they should be implemented by a cost function-based market maker satisfying certain conditions in order to satisfy certain information elicitation axioms.
We define these below.
\begin{definition}[Cost function-based market maker]
  A cost function-based market maker with $n$ securities is one that prices each security via a differentiable potential function $C:\mathbb{R}^n\to\mathbb{R}$.
  Suppose a trader wants to purchase a bundle of securities $\r\in\mathbb{R}^n_{\geq 0}$; that is, $r_i$ shares of security $i$, when the market has current liability of $\q$.
  Then, the trader must pay $C(\q+\r)-C(\q)$ in cash to the market maker.
\end{definition}

Cost function-based markets always maintain a \emph{liability} $\q\in\mathbb{R}^n$ of securities, $q_i$ of security $i$ that the market has sold so far.
The term liability comes from the fact that $q_i$ is the amount due to traders upon outcome $i$.
As shown by \citet{abernethy2013efficient}, prediction markets that elicit information well are precisely cost function-based prediction markets with $C$ convex and $\ones$-invariant. 

\begin{remark*}\label{remark:no-cash}
  As remarked in \citet{frongillo2023axiomatic}, in a complete AD securities market, holding one of each security, i.e., $\ones$, is equivalent to holding \$1 cash.
  Without loss of generality, therefore, one can restrict attention to ``net trades'' $\r=\r'-(C(\q+\r')-C(\q))\ones$, which subtract the cost of the trade $\r'$, converting cash to $\ones$.
  Using the $\ones$-invariant property of $C$, we can see that $C(\q+\r) = C(\q)$.
  While it is not customary for predictions to deal with the net trade, in our setting net trades simplify various results, and more readily connect to the decentralized finance literature.
\end{remark*}

\subsection{Scoring rules} \label{sec:scoring-rules}

Scoring rules were introduced by \citet{brier1950verification} to score forecasts of a random variable such as $Y$ above.
In this setting, we seek to design a score $S(\p,y)$ that determines the quality of prediction $\p \in \Delta_\Y$ upon observing the outcome $y\in\Y$, with the property that $\E_\p S(\hat{\p},Y)$ is maximized at $\hat{\p} = \p$.
The full characterization of such ``proper'' scoring rules takes the form
\begin{align*}
    S_G(\p,y) = G(\p) + \inprod{\dG_\p}{\delta_y - \p}~
\end{align*}
where $G:\Delta_\Y\to\reals$ is a convex function~\citep{gneiting2007strictly}.
When $\Y = \{0,1\}$, we can write $p\in[0,1]$ to be the predicted probability that $Y=1$, and write $S_g(p,y) = g(p) + g'(p)(y-p)$ for $g:[0,1]\to\reals$ convex.

\citet{hanson2003combinatorial} showed that scoring rules could be used to design AMMs in a form we call a \emph{scoring rule market}; see Protocol~\ref{alg:general-n-asset-dual-price}.
The basic idea is to pay traders according to a difference of scoring rules, with the latest trader's prediction acting as the current market price.
It was later shown that this formulation is equivalent in a strong sense to the cost function market makers described above~\citep{abernethy2013efficient}.
Specifically, the scoring rule market for $S_G$ is equivalent to the cost function market maker with cost function $C = G^*$, the convex conjugate of $G$.

A corollary of these connections, leveraged in~\citet{frongillo2023axiomatic}, is that one can use scoring rules as vectors to convert between the market price vector and the current liability/reserve vector.
Specifically, let $S_G(\p,\cdot) = (S_G(\p,y))_{y\in\Y} \in (\reals\cup\{\infty\})^n$ be the scoring rule vector for price $\p$.
Then, up to a uniform shift, the liability vector of a cost function market maker with cost function $C = G^*$ at price $\p$ is simply $S_G(\p,\cdot)$.
We will use this correspondence throughout the paper, as well as the 2-outcome version
$S_g(p,\cdot) = (g'(p),0) + (g(p)-p\cdot g'(p))\ones \in (\reals\cup\{\infty\})^2$.

\subsection{Technical definitions} \label{sec:techn-cond-lemm}

We give some technical definitions relating to \S~\ref{sec:predmarketscostfuncs} and \S~\ref{sec:scoring-rules} that we use in later sections.
\begin{definition}[Smoothness of $G$]
  We say a convex function $G:\Delta_n \to \reals$ is \emph{smooth} if its 1-homogeneous extension $\overline G : \reals^n_{\geq 0}$ is differentiable.
\end{definition}

The key class of the functions $G$ we restrict to is as follows.
\begin{definition}[Generating function]
  We say $G:\Delta_n \to \reals$ is a \emph{generating function} if it is
  convex, smooth, and bounded on $\Delta_n$.
  
\end{definition}
\begin{definition}[Pseudobarrier, \citet{abernethy2013efficient}]
  A generating function $G$ is a \emph{pseudobarrier} if for any sequence $\{\p^j\in\relint\Delta_n\}_j$ converging to the relative boundary of $\Delta_n$, and $\{\q^j \in \partial G(\p^j)\}_j$, then $\|\q^j\| \to \infty$.
\end{definition}
A common example of a pseudobarrier\footnote{This term was coined in \citet{abernethy2013efficient} and used similarly to our setting: ensuring that the market price remains in the relative interior of the simplex.} is (negative) Shannon entropy $G(\p) = \sum_y p_y \log p_y$.
Another is the dual of the constant product market maker $G(\p) = -n(\prod_y p_y)^{1/n}$.\footnote{To see that this dual is correct, one can observe that it is 1-homogeneous, and thus $S(\p,y) = -(\prod_{y'} p_{y'})^{1-1/n} (\prod_{y'\neq y} p_y)$.  Now letting $\res = -S(\p,\cdot)$ be the corresponding reserve vector, and computing the product, we have $\prod_y q_y = \prod_y (\prod_{y'} p_{y'})^{1-1/n} (\prod_{y'\neq y} p_y) = 1$.}

We now give the sets of cost and generating functions used in the general protocols.
Let $\G_n$ be the set of nonpositive generating functions $G:\Delta_n\to\reals_{\leq 0}$, and $\G_n^* \subseteq \G_n$ those which are pseudobarriers.\footnote{Given any bounded generating function $G$, to obtain the optimal liability, we can simply replace it by the function $\p \mapsto G(\p) - \inprod{\p}{\q}$ where $q_i = G(\indvec{i})$.}
Let $\C_n$ and $\C_n^*$ be the sets of conjugates of $\G_n$ and $\G_n^*$, respectively.

\subsection{Automated market makers and liquidity provisioning in decentralized finance}\label{sec:amms}

A major AMM innovation in decentralized finance is the introduction of \emph{liquidity provisioning}.
In traditional AMMs, the market maker takes on an additional risk of price fluctuations of the reserves for the ability to run a market always willing to price a bundle of assets.
The decentralized finance implementations of AMMs, though, have outsourced provisioning these reserves, and hence liquidity, to external parties called liquidity providers (LPs).
The AMMs typically define trade dynamics when liquidity is fixed.
In this setting, traders can exchange assets with the market maker in a way that keeps the reserves/liability on the same invariant curve of $\varphi$ or $C$.
Decentralized finance protocols like Uniswap V2 and Uniswap V3 also allow liquidity providers (LPs) to change the market's liquidity while keeping the price $\p$ invariant \citep{adams2020whitepaperv2,adams2021whitepaperv3}.
LPs may either add, or \textit{mint}, liquidity to the market or remove, or \textit{burn}, liquidity from the market. 
LPs make it easier for the AMM to conduct trades by absolving the market maker of the risk of providing liquidity.
As compensation for taking on the risk, LPs are rewarded using trading fees, which are skimmed off along with the trade requested.
These fees form a pool to be distributed proportionally to LPs as the liquidity they allocated is used.

Notably, in decentralized finance, the goal is not information elicitation, as in prediction markets, but rather the exchange of assets---namely cryptocurrencies.
Analogous to how cost function markets maintain a liability vector $\q$, AMMs maintain a reserve vector $\x=-\q$ of assets.
Constant function market makers (CFMMs) are a special type of AMM.
For various restrictions on their design, \citet{frongillo2023axiomatic,angeris2020improved,schlegel2022axioms} argue that CFMMs satisfy desirable market making axioms.
We define these market makers below.
\begin{definition}[Constant function market maker, CFMM]
  A constant function market maker (CFMM) is a market maker based on a potential function $\varphi:\mathbb{R}^n\to\mathbb{R}$ that maintains a liability $\q\in\mathbb{R}^n$.
  At the current liability, the set of trades $\r$ available are those that satisfy $\varphi(\q+\r)=\varphi(\q)$.
  After a trade, the liability vector updates to $\q\gets\q+\r$.
\end{definition}
Consistent with this definition, in this paper, we adopt the sign convention that trades are always oriented toward the trader.
For example, a trade $\r = (1,-3)$ corresponds to giving the trader 1 unit of asset 1 in exchange for 3 units of asset 2

The relationship between cost function prediction market makers and CFMMs is thoroughly explored in \citet{frongillo2023axiomatic}. 
The two objects are different, but give rise to equivalent characterizations of markets.
The classic cost function market makers commonly employed in prediction markets are special cases of CFMMs that retain the full flexibility of general potential functions $\varphi$.
We use cost functions throughout, even when discussing CFMMs, as they have more mathematical structure without loss of expressiveness.
For example, while for any potential $\varphi$, one can take ratios of partial derivatives to compute relative prices, this task is even easier for cost functions, as prices are normalized.
In the context of prediction markets, normalization means that prices $\p \in \partial C(\q)$ can be thought of as a probability distribution over outcomes.
With or without this interpretation, we frequently use the fact that $\partial C(\q) \subseteq \Delta_n$ for all $\q\in\reals^n$.

\section{Liquidity provisioning protocol for prediction markets}
\label{sec:protocol-pred-markets}

In this section, we first give intuition in \S~\ref{subsec:competing-markets} for why liquidity provisioning can be thought of as running ``parallel'' markets.
We detail our general, cost function-based, prediction market liquidity provisioning protocol in \S~\ref{sec:protocol_main}.
Then, in \S~\ref{sec:protocol_msr} we provide a scoring rule-based protocol.
We show both protocols to be equivalent in \S~\ref{sec:equivalence}, and we also show how these protocols lead to several equivalent ways of thinking about LPs running markets in parallel.
These intuitive interpretations not only enhance our understanding of how liquidity provisioning can be implemented in prediction markets, but also justify the framework we propose.
We end this section by discussing some constraints on the design of our protocol's market-making functions in \S~\ref{sec:no_liability} and \S~\ref{sec:optimal_deposits}.

\subsection{Liquidity provisioning as competing market makers}\label{subsec:competing-markets}

In traditional financial markets, such as continuous double-auctions, a market maker is an entity that offers both to buy and sell an asset.
Typically the buy price is lower than the sell price; the difference comprises the \emph{bid-ask spread}.
Market makers earn a profit equal to the bid-ask spread whenever both buy and sell orders are executed, while remaining even with respect to the asset.
In essence, market making is all about providing liquidity for a small premium, or ``fee,'' as given by the spread.
In these traditional markets, liquidity provisioning happens naturally, as often multiple market makers coexist.
Rational traders will only buy or sell from the most favorable price offered, switching at will between different market makers.

The key idea behind our protocol, therefore, is to implement liquidity provisioning in the same manner, with multiple coexisting automated market makers.
That is, we seek a protocol that implements an LP as simply another ``competing'' market maker, and we let traders interact with them all at once, i.e., in parallel.
How could one implement such a protocol?

We will show that there are several equivalent ways to imagine this transaction proceeding.
First, a trader could select a valid trade $\r^i$ for the automated market maker of each LP $i$, resulting in a net trade $\r = \sum_i \r^i$.
As detailed in \S~\ref{sec:amms}, these valid trades can be expressed as those satisfying $C_i(\q^i+\r^i) = C_i(\q^i)$ given a convex cost function $C_i$ and current liability vector $\q^i$ for LP $i$.
Second, a trader could execute a trade in continuous time, at each moment trading with the LP offering the most favorable price, eventually stopping and yielding a net trade $\r$.
Perhaps surprisingly, by fundamental results in convex analysis, these two approaches are identical.
Taken together, we can see that any Pareto-optimal trade leaves the combined market in a coherent state, with the price of each LP matching the global market price.

At first glance, it might appear that a major downside of our approach is the need for traders to interact directly with each LP, increasing the complexity of interaction required.
Fortunately, one can simplify the interface: there always exists a single aggregate cost function $C$ that captures the available net trades.
Specifically, given the cost functions $C_i$ defining each market maker, the valid trades in the combined market are exactly those of their \emph{infimal convolution} $C = \bigwedge_i C_i$.
Thus, the trader can simply choose any trade satisfying $C(\q+\r)=C(\q)$, and behind the scenes, the split $\r = \sum_i\r^i$ can be computed along with the corresponding fees.
This aggregation yields a third equivalent notion of competing market makers.

Equivalently, we may consider the scoring rule market (SRM) formulation of cost functions.
As mentioned in \S~\ref{sec:scoring-rules}, a cost function market given by a convex cost function $C$ is equivalent to the scoring rule market given by its convex conjugate $G=C^*$.
We show that our notion of LPs as parallel market makers extends quite elegantly to SRMs.
The combined market simply uses the sum of the scoring rules $S_{G_i}$ of all the component LPs, or equivalently uses a single scoring rule $S_G$ generated by the sum of the generating functions $G = \sum_i G_i$.
See \S~\ref{sec:protocol_msr} for the full details of the scoring rule version of our protocol.

In \S~\ref{sec:liquidity_discussion}, we discuss how liquidity can be expressed as a matrix-valued function, either by  $(\nabla^2 C)^+$, or dually, $\nabla^2 \overline G$ where $G = C^*$.
As a consequence, the total liquidity of the combined market is $\nabla^2 \overline G = \sum_i \nabla^2 \overline G_i$.
In other words, just as one would hope, adding another parallel market maker simply adds the corresponding liquidity $\nabla^2 \overline G_i$ to the pool.

\begin{remark*}
  Many existing protocols require LPs to simply deposit tokens to provide liquidity. 
  At first glance, it might appear that requiring LPs to provide cost functions $C_i$ is more complex.
  Our Protocol~\ref{alg:general-n-asset-cost}, and its equivalent protocols, clarifies that even in existing protocols that just ask LPs to deposit tokens, the LPs are implicitly specifying cost functions.
  Hence, we may recover existing protocols as special cases by specifying a restricted set $\CLP$ of cost functions that LPs may (explicitly or more commonly, implicitly) provide.
  For instance, in Protocol~\ref{alg:V2-two-asset}, we recover Uniswap V2, where LPs simply deposit some numerical quantity of ``liquidity'' that is spread evenly across the entire price space.
  Under the hood, LPs are specifying scaled copies of the ``base shape'' of Uniswap's cost function.
  In the more expressive Uniswap V3 which we recover in Protocol~\ref{alg:V3-two-asset}, LPs deposit tokens determined by a numerical quantity of ``liquidity'' in discrete price ranges called buckets, and again they are implicitly specifying cost functions.
\end{remark*}

\subsection{Cost function-based general protocol} \label{sec:protocol_main}

\begin{algorithm}[t]
  \caption{General protocol as parallel market makers}
  
  \label{alg:general-n-asset-cost}
  \begin{algorithmic}[1]
    \State \textbf{global constant} $\Cinit$, $\CLP$, $\fee(),\fee_i()\in\reals_{\geq 0}$.
    \State \textbf{global variables} $k\in\N$, $\{\q^i\in\reals^n\}_{i=0}^k$, $\{C_i\in\CLP\}_{i=0}^k$
    
    \State $\liabilityf(C) := \q\in\reals^n$ s.t.\ $\partial C_0(\q^0) \cap \partial C(\q)\neq\emptyset$ and $C(\q) = 0$ \Comment{Price matching, no-liability}
    \label{item:general-n-asset-cost-liabilityf}

    \medskip
    
    \Function{Initialize}{$\q \in \reals^n,C\in\Cinit$}
    \State $(k,\q^0,C_0) \gets (0,\q,C)$
    \State \textbf{check} $\q^0 = \liabilityf(C_0)$ s.t. no-liability and optimal deposits are satisfied.
    \EndFunction
    \medskip

    \Function{RegisterLP}{$i=k+1$}
    \State $(k,\q^i,C_i) \gets (k+1,0,\max)$
    \EndFunction
    \medskip

    \Function{ModifyLiquidity}{$i \in \N, C \in \CLP$}
    \State \textbf{request} $\r^i = \q^i - \liabilityf(C)$ from LP $i$ s.t. no-liability and optimal deposits are satisfied.
    \State $(\q^i,C_i) \gets (\q^i - \r^i,C)$
    \EndFunction
    \medskip
    
    \Function{ExecuteTrade}{$\r \in \reals^n$}

    \State $\q \gets \sum_{i=0}^k \q^i$
    \State \textbf{check} $C(\q + \r) = C(\q)$ where $C = \bigwedge_{i=0}^k C_i$
    \label{item:general-n-asset-cost-trade-check}
    \State \textbf{trader pays} $\fee(\r,\q)$ cash in fees
    \State \textbf{give} $\r$ to trader
    \State write $\r = \sum_{i=0}^k \r^i$ s.t.\ $\forall i$, $C_i(\q^i+\r^i) = C_i(\q^i)$
    
    \label{item:general-n-asset-cost-trade-split}
    \For{each LP $i$}
    \State LP $i$ gets $\fee_i(\r,\q)$ fees
    \State $\q^i \gets \q^i + \r^i$
    \EndFor
    \EndFunction
  \end{algorithmic}
\end{algorithm}

Our proposed protocol for liquidity provisioning in prediction markets is described in Protocol~\ref{alg:general-n-asset-cost}.
At a high level, the protocol works as follows.
Let $n$ be the number of securities.
The market creator acts as the initial LP, giving reserves $\q^0$ to the liquidity pool and specifying the initial cost function $C_0$.
When an additional LP $i$ enters, their liability vector and cost function are initialized to the trivial values $\q^i=\0$ and $C_i(\q)=\max(\q) := \max_j q_j$, so that they initially provide no liquidity.\footnote{To see why this choice is correct, note that the ``bid-ask spread'' of $C_i$ at $\q^i=0$ is maximal; every price vector in $\Delta_n$ is consistent, and any trade occurs at the worst feasible price.  More technically, adding the max function to the infimal convolution $C= \bigwedge_{i=0}^k C_i$ does not change the result.  Dually, the conjugate of $\max$ is the convex indicator of $\Delta_n$, so this choice adds liquidity $G_i = 0$; see \S~\ref{sec:protocol-pred-markets}.}
The \textsf{ModifyLiquidity} function handles an LP adding, removing, or otherwise altering their deposited liquidity: they simply replace their cost function with a different one, and are charged up-front the minimal deposit to ensure \emph{no liability}, i.e., that they will never owe the market maker in any future state.
We provide more discussion of the \emph{no-liability} condition in \S~\ref{sec:no_liability}.
When removing all liquidity, the LP simply sets $C_i = \max$ once again, and is given back their entire deposit.
\textsf{ExecuteTrade} checks if a trade $\r$ is an allowed trade with the overall cost function $C= \bigwedge_{i=0}^k C_i$, and if so, requests an additional fee of $\fee(\r,\q)$ cash from the trader.
Under the hood, it then finds the optimal split $\r = \sum_i \r^i$ into smaller trades, executing each with the corresponding LP and doling out $\fee_i(\r,\q)$ in fees.

Two reasonable fee choices are $\fee(\r,\q) = \beta\|\r\|$ and $\fee_i(\r,\q) = \beta\|\r\|\frac{\|\r^i\|}{\sum_j\|\r^j\|}$ for some norm $\|\cdot\|$ and $\beta>0$.
The form of $\fee_i$ here is to ensure budget balance of the fees, so that the market maker does not owe LPs more than the trader pays.
We provide a deeper discussion of the design of the fee functions in \S~\ref{sec:fees}, and prove that natural axioms are incompatible.

A key step in the protocol is the computation of liabilities, from which several questions arise.
If an LP wishes to provide liquidity using $C$, and the current price is $\p$, what do they need to deposit?
Also, does the required split $\r=\sum_i \r^i$ always exist? 
We answer these questions in the later sections.

\subsection{Scoring rule-based general protocol} \label{sec:protocol_msr}

As some readers might be more familiar with the  scoring rule markets of \citet{hanson2003combinatorial}, we give Protocol~\ref{alg:general-n-asset-dual-price} as an alternative to Protocol~\ref{alg:general-n-asset-cost}.
Here, we keep track of the market price $\p$ instead of the liability vector $\q$ and use scoring rules instead of cost functions.
We prove in Appendix~\ref{sec:equiv-full-prot} that Protocols~\ref{alg:general-n-asset-cost} and ~\ref{alg:general-n-asset-dual-price} are equivalent under mild conditions.

\begin{algorithm}
  \caption{General protocol as parallel scoring rule markets}
  \label{alg:general-n-asset-dual-price}
  \begin{algorithmic}[1]
    \State \textbf{global constant} $\Ginit$, $\GLP$, $\fee(),\fee_i()\in\reals_{\geq 0}$.
    \State \textbf{global variables} $k\in\N$, $\p\in\relint\Delta_n$, $\{G_i\in\GLP\}_{i=0}^k$

    \State $\liabilityf(G,\p) := S_G(\p,\cdot)$ \hfill where $ S_G(\p,\cdot) = \dG_\p + (G(\p) - \inprod{\dG_\p}{\p})\ones$.
    \label{item:MSR_protocol_liabilitysf}
    
    \medskip
    \Function{Initialize}{$\p \in \relint\Delta_n,G\in\Ginit$}
    \State $(k,\p,G_0) \gets (0,\p,G)$
    \State $\q^0 := -S_{G_0}(\p,\cdot)$ s.t. no-liability and optimal deposits are satisfied.
    \EndFunction
    \medskip

    \Function{RegisterLP}{$i = k+1$}
    \State $(k,G_k) \gets (k+1,0)$
    \EndFunction
    \medskip

    \Function{ModifyLiquidity}{$i \in \N, G' \in \GLP$}
    \State \textbf{request} $S_{G_i}(\p,\cdot) - S_{G'}(\p,\cdot)$ from LP $i$ s.t. no-liability and optimal deposits are satisfied.
    \State $G_i \gets G'$
    \EndFunction
    \medskip

    \Function{ExecuteTrade}{$\hat\p \in \relint\Delta_n$}

    \State $\r \gets S_G(\hat \p,\cdot) - S_G(\p,\cdot)$ where $G = \sum_i G_i$ 
    \State \textbf{pay} $\fee(\r,\p)$ cash in fees
    \State give $\r$ to trader
    \For{each LP $i$}
    \State $\r^i \gets  S_{G_i}(\hat \p,\cdot) - S_{G_i}(\p,\cdot)$.
    \State \textbf{pay} $\fee_i(\r,\p)$ in fees to LP $i$
    \State $\q^i \gets \q^i + \r^i$
    \EndFor
    \EndFunction
  \end{algorithmic}
\end{algorithm}

\subsection{Ensuring no liability}\label{sec:no_liability}

To capture the \emph{no-liability condition}, we require that each LP $i$ should never owe the market maker shares of any security.
When entering the market, LP $i$ is required to deposit some assets $\x\in\reals^n_{\geq 0}$ according to its cost function $C_i$ and the current overall liability.
We consider LP $i$'s liability to be $\q^i=-\x$ where $\q^i$ is chosen so that $C_i(\q)$ implies $\q \preceq 0$ for all possible future states $\q$.

One way to ensure this is that, if an LP provides a cost function $C$ where some level set satisfies no-liability, the protocol could compute it and request a corresponding deposit to cover the liability.
More formally, if $\{\q\in\reals^n \mid C_i(\q) = \alpha\} \subseteq \reals_{\leq 0}^n$, i.e.\ every liability vector in the $\alpha$-level set of $C_i$ is nonpositive, the protocol could compute this $\alpha$ and request a deposit $-\q$ in the $\alpha$-level set such that $\nabla C_i(\q) = \p$, the current price.
Equivalently, we may require that the LP encode the valid trades in the zero level set by submitting an offset cost function $\hat C_i := C_i - \alpha$, so that the $\alpha$-level set is shifted to have value 0.
In Protocol~\ref{alg:general-n-asset-cost}, we require LPs submit such a cost function.
Equivalently, we may ensure no-liability in Protocol \ref{alg:general-n-asset-dual-price} by requiring the generating function $G_i$ to be nonpositive everywhere on the simplex.

Recall that for two outcomes, the constant product potential $\varphi(\q) = q_1 q_2$ has an equivalent cost function \citep{chen2007utility,frongillo2023axiomatic} given by $$C(\q) = \frac 1 2 \left( q_1 + q_2 + \sqrt{4\alpha^2 + (q_1-q_2)^2} \right).$$
Taking the $0$-level set, we have
$q_1 + q_2 + \sqrt{4\alpha^2 + (q_1-q_2)^2} = 0$ which implies $q_1 + q_2 < 0$.
Then $(q_1 + q_2)^2 = 4\alpha^2 + (q_1-q_2)^2$, which reduces to $q_1 q_2 = \alpha^2$.
By the first observation, we must have $\q \prec \0$, giving no liability.
While one could take $\hat C(\q) = C(\q) + 1$ to arrive at the same liquidity level, one can check that $C$ gives rise to the minimal deposit required for that level.

\subsection{Optimal deposits} \label{sec:optimal_deposits}

While the no-liability condition requires that an LP's liability always be nonnegative, the optimal deposits condition requires that all the assets deposited by some LP may be used in some trade, and therefore none of it is ``wasted.''
To satisfy optimal deposits, the LP must deposit a liability $\q$ such that for every asset $i$ and any $\epsilon>0$, there exists a valid trade $\r$ such that the resulting liability $\hat{\q}=\q+\r$ contains less than $\epsilon$ securities\footnote{written as $\hat{\q}_i>\epsilon$ for any $\epsilon<0$, as liabilities are negative.} of asset $i$
We can characterize optimal deposits in the SRM protocol (Protocol~\ref{alg:general-n-asset-dual-price}) by requiring that the generating function $G$ approaches 0 at the vertices of the simplex.
\begin{proposition}
Let $G:\Delta_n\to\mathbb{R}$ be a convex generating function where $\lim_{j\to\infty}G(\p^j)=0$ for any Cauchy sequence $(\p^j)_{j=0}^\infty$ of prices where $\lim_{j\to\infty}$ $\p^j=\hat{\p}$ where $\hat{\p}$ is the unique price with $\hat{\p}_i=1$.
Given any price $\p$ and corresponding liability vector $\q=S_G(\p,\cdot)$, for any $\epsilon<0$, there exists a legal trade $\r$ at liability $\q$ such that $q_i+r_i>\epsilon$.
\end{proposition}
\begin{proof}
Given an initial liability $\q$, legal trade $\r$, and resulting liability $\hat{\q}$, recall from Protocol~\ref{alg:general-n-asset-dual-price} that
\(
\hat{\q}=S_G(\hat{\p},\cdot) = \dG_{\hat{\p}} + (G(\hat{\p}) - \inprod{\dG_{\hat{\p}}}{\hat{\p}})\ones.
\)
The component $\hat{q}_i$ of $\hat{\q}$ corresponding to security $i$ is therefore
\(
S_G(\hat{\p},y_i)=\dG_{\hat{p}_i}+G(\p)-\dG_{\hat{p}_i}=G(\hat{\p})
,\)
where $y_i$ is the outcome corresponding to asset $i$ and $\dG_{\hat{p}_i}$ is the $i^{th}$ component of some subgradient $\dG_{\hat{\p}}$ of $G$ at $\hat{\p}$.
Now, let $\hat{\p}$ be the unique price such that $\hat{p}_i=1$. Let $(\p^j)_{j=0}^\infty$ be some Cauchy\footnote{with respect to a norm on the price space.} sequence of prices approaching $\hat{\p}$, and $(\q^j)$ be the corresponding sequence of liability vectors.
We have that
\[
\lim_{j\to\infty} \q^j = \lim_{j\to\infty}\ S_G(\p^j,y_i) = \lim_{j\to\infty}\dG_{p^j_i}+G(\p^j)-\dG_{p^j_i} = \lim_{j\to\infty} G(\p^j)=0.
\]
As the limit approaches 0 from below, for any $\epsilon<0$, there exists some $\p^j$ such that $q^j_i>\epsilon$.
\end{proof}

\subsection{Example of equivalent protocols in action} \label{sec:example_run}
We construct an example of both equivalent protocols in action that makes use of the Logarithmic Market Scoring Rule (LMSR), popularized by \citet{hanson2003combinatorial}.
We consider a setting with three assets and two LPs.
While we focus on the SRM version of our protocol for the sake of easy and clean computation, it is simple to derive the exchange of securities and change in liabilities from price movement simply by computing the value scoring rule $S_G(\p,\cdot)$.
We thus interleave both interpretations at once.
Consider the following generating functions:
\begin{align*}
& G_1(\p) = p_1\log(p_1)+p_2\log(p_2)-(p_1+p_2)\log(p_1+p_2), \\
& G_2(\p) = p_2\log(p_2) + p_3\log(p_3)-(p_2+p_3)\log(p_2+p_3).
\end{align*}
Both functions are 1-homogeneous versions of the LMSR defined on two assets.
Their corresponding scoring rules are
\begin{align*}
& S_{G_1}(\p,\cdot) = \nabla G_1 = \left(\log\left(\frac{p_1}{p_1+p_2}\right),\log\left(\frac{p_2}{p_1+p_2}\right),0\right), \\
& S_{G_2}(\p,\cdot) = \nabla G_2 = \left(0, \log\left(\frac{p_2}{p_2+p_3}\right),\log\left(\frac{p_3}{p_2+p_3}\right)\right).
\end{align*}

Suppose we start at an initial price vector of $\p=(1/3,1/3,1/3)$, and trade to a price of $\hat{\p}=(1/7,2/7,4/7).$
Both LPs must make an initial deposit equal to $S_{G_i}(\p,\cdot)$:
\begin{align*}
\q^1 = \nabla G_1(\p) = (\log(1/2),\log(1/2),0), \quad
 \q^2 = \nabla G_2(\p) = (0,\log(1/2),\log(1/2)).
\end{align*}
For each LP $i$, we calculate the trade $\r^i$ as $S_{G_i}(\hat{\p},\cdot)-S_{G_i}(\p,\cdot)$, and then set a new liability for LP$_i$ of $\hat{\q}^i=\q^i+\r^i$.
Equivalently, we can calculate $\hat{\q}^i=S_{G_i}(\hat{\p},\cdot)$, and then calculate $\r^i$ as the difference in liabilities $\hat{\q}^i-\q^i$.
At the new price of $\hat{\p}=(1/7,2/7,4/7)$, the respective liabilities are
\begin{align*}
\hat{\q}^1 = \nabla G_1(\hat{\p}) = (\log(1/3),\log(2/3),0), \quad \hat{\q}^2 = \nabla G_2(\hat{\p}) = (0, \log(1/3),\log(2/3)).
\end{align*}
And the net trades $\r^1,\r^2$ are therefore:
\begin{align*}
& \r^1 = \hat{\q}^1 - \q^1 =  (\log(1/3),\log(2/3),0)-(\log(1/2),\log(1/2),0) = (\log(2/3),\log(4/3),0), \\
& \r^2 = \hat{\q}^2 - \q^2 = (0, \log(1/3),\log(2/3)) - (0,\log(1/2),\log(1/2)) = (0, \log(2/3),\log(4/3)).
\end{align*}
With some fee $\fee_i(\r^i,\p)$ assessed for each LP $i$.

Note that these generating functions exhibit an intuitive property: they make a deposit of 0 and facilitate a trade of 0 in the asset they are not parameterized by.
The same holds for 1-homogeneous generating functions in general.
When we discuss matrix-valued liquidity in \S~\ref{sec:liquidity-hessian}, we will note that these functions also do not provide any liquidity for trading in those assets.

Also observe that this LP configuration is exactly equivalent to having a single LP with the generating function $G=G_1+G_2$.
For the same initial and final prices, $S_G(\p,\cdot)=S_{G_1}(\p,\cdot)+S_{G_2}(\p,\cdot)$, and accordingly the LP's initial liability will be $\q=\q^1+\q^2$, its final liability will be $\hat{\q} = \hat{\q}^1+\hat{\q}^2$, and the net trade will be $\r=\r^1+\r^2$.

We now consider the convex conjugates $C_1=G^*_1$ and $C_2=G^*_2$, and show that the cost function protocol yields an equivalent result to the scoring rule protocol. We have
\begin{align*}
C_1(\q) = \max\{\log(e^{q_1}+e^{q_2},q_3)\},\quad 
C_2(\q) = \max\{\log(e^{q_2}+e^{q_3},q_1)\}.
\end{align*}
While one might expect $C_1$ and $C_2$ to be parameterized by the same 2 assets as their conjugates, the maximization term has a natural interpretation---purchasing more than a small amount of an asset for which an LP provides no liquidity will cause its price to reach 1.
We can verify that the trades described above comport with the cost function protocol.
Observe that
\begin{align*}
C_1(\q^1+\r^1)-C_1(\q^1)&=C_1(\hat{\q}^1)-C_1(\q^1)=\log(e^{\log(1/3)}+e^{\log(2/3)})-\log(e^{\log(1/2)}+e^{\log(1/2)}) \\
&=\log(1/3+2/3)-\log(1/2+1/2)=0.
\end{align*}
As $C(\q^1+\r^1)-C(\q^1)=0$, $\r^1$ is a legal trade at liability $\q^1$ Both $C_2$ and the infimal convolution $C=C_1 \wedge C_2=(G_1+G_2)^*$ can be checked similarly.

\section{Equivalence of interpretations}\label{sec:equivalence}

As we have argued informally, one can regard liquidity provisioning as (a) recruiting multiple market makers, which then (b) process trades in parallel.
We now study (b) formally, showing that five natural ways to interpret this parallelism are all equivalent.
The equivalence of (a) in these interpretations, the process of recruiting market makers and securing deposits, is then straightforward (\S~\ref{sec:equiv-full-prot}).

\subsection{Five interpretations of parallel market making}\label{sec:interpretations}

To begin, we revisit the interpretation of liquidity provisioning as running several market makers in parallel, and show that five natural interpretations of this idea lead to equivalent protocols.

In interpretations 1, 2, and 5, each market maker $i$ is specified by a cost function $C_i$ and a state $\q^i$.
In interpretations 3 and 4, market makers instead each have a scoring rule $S_i$ generated by a convex function $G_i=C_i^*$ and maintain a price $\p^i$.
In all cases, we assume the trader behaves rationally, in the sense that the overall trade is Pareto optimal: if $\r,\r'$ are both valid trades, and $\r \succeq \r'$, the trader chooses $\r$.
Recall that trades are oriented toward the trader, so here $\r$ gives the trader weakly more of each security.

In each interpretation, we capture the market state by the collection of liability vectors $\{\q^i \in \reals^n\}_i$.
After a trade $\r = \sum_i \r^i$, the state updates to $\{\q^i + \r^i\}_i$.
The set of consistent market prices is defined to be $\partial C_i(\q^i)$ for interpretations 1, 2, and 5, and an analogous definition for 3 and 4.
We say the overall market state is \emph{coherent} if there is a consistent price $\p\in\relint\Delta_n$ for all market makers simultaneously.

\begin{enumerate}
\item \textbf{The trader selects a valid trade from each market maker's cost function and executes them all.}
  Formally, for each market maker $i$ the trader selects $\r^i$ such that $C_i(\q^i+\r^i) = C_i(\q^i)$, for a total trade of $\r = \sum_i \r^i$.
\item \textbf{A centralized market maker uses the infimal convolution of cost functions.}
  This interpretation corresponds to the rules for trade in Protocol~\ref{alg:general-n-asset-cost}.
  Formally, the trader selects any $\r\in\reals^n$ such that $C(\q + \r) = C(\q)$, where $C = \bigwedge_i C_i$ and $\q = \sum_i \q^i$.
  The central market maker first gives $\r$ to the trader.
  Behind the scenes, it then computes a split $\r = \sum_i \r^i$ such that $C_i(\q^i + \r^i) = C_i(\q^i)$, whose existence we establish below, and executes these trades in each constituent market maker.
\item \textbf{The trader is paid according to the sum of each market maker's scoring rule.}
  Formally, each market maker has a scoring rule $S_i(\p,y) = G_i(\p) + \inprod{\dG_\p}{\indvec{y} - \p}$ where $G_i = C_i^*$ and $\{\vec d G_\p\in\partial G(\p) \mid \p\in\relint\Delta_n\}$ is an arbitrary selection of subgradients.
  Market maker $i$ maintains a price vector $\p^i \in \relint\Delta_n$, and the trader may choose any $\hat\p^i \in \relint\Delta_n$, resulting in the trade $\r^i = S_i(\hat\p^i,\cdot) - S_i(\p^i,\cdot) \in \reals^n$.
  See Protocol~\ref{alg:general-n-asset-dual-price}.
  For the purposes of comparing interpretations, 
  define the set of consistent prices as $\{\p \in \relint\Delta_n \mid \forall i\, S_i(\p,\cdot) = S_i(\p^i,\cdot)\}$.
\item 
  \textbf{A centralized market maker chooses the generating function $G=\sum_{i=1}^k G_i$ and the corresponding scoring rule $S_G=\sum_{i=1}^k S_{G_i}$.}
  The market maker maintains a price vector $\p \in \relint\Delta_n$, and the trader may choose any $\hat{\p} \in \relint\Delta_n$, resulting in the trade $\r = S_G(\hat{\p},\cdot) - S_G(\p,\cdot) \in \reals^n$.
\item \textbf{The trader continuously trades at the most favorable price and at some point stops.}
  Recall that we can interpret a cost function $C_i$ as quoting a cost $C_i(\q^i + \v^i) - C_i(\q^i)$ for each bundle of assets/securities $\v^i\in\reals^n$.
  Formally, in this interpretation, the trader specifies a direction $\v\in\reals^n$, and a stopping point $\alpha$, and continuously purchases $\v dt$ for the smallest price $C'_i(\q^i;\v)$ over all $i$, for $\alpha$ units of time.
  Here $C'_i(\q^i;\v) := \lim_{h\to 0^+} \frac{C_i(\q^i+h\v) - C_i(\q^i)}{h}$ is the directional derivative of $C_i$.
  In other words, the trades are of the form $(\v - C_i'(\q^i;\v)\ones)dt$, where we recall that the numeraire is simply $\ones$.
  Crucially, we also allow the trader to take advantage of any arbitrage opportunity that arises from this continuous trade: after the $\alpha$ units of time, the trader may place trades $\{\hat \r^i\}_i$ if they have negative net cost and $\sum_i \hat\r^i = \0$.
\end{enumerate}

\subsection{Equivalence of the interpretations}
\label{sec:equiv-interpr}

Before stating the equivalence of these interpretations, we must address an important technical point.
Recall the definitions in \S~\ref{sec:techn-cond-lemm}.
By our assumptions on $G$, the resulting scoring rule vectors are unique for each price $\p\in\relint\Delta_n$; see Lemma~\ref{lem:cost-scoring-duality} in \S~\ref{sec:app-proofs-of-equivalence}.
But, on the relative boundary of the simplex $\Delta_n$, uniqueness can fail.
If we take $G(\p) = \|\p\|_2^2$, or any LP that does not provide infinite liquidity at the boundary of $\Delta_n$, the resulting $G$ will not have a unique subgradient (even modulo $\ones$) at those boundary points, meaning we will have $\q,\hat{\q} \in \reals^n$ with $\p \in \partial C(\q)\cap\partial C(\hat{\q})$, but with $\hat{\q} - \q \neq \alpha\ones$ for any $\alpha$.
The scoring rule $S$ must pick just one of these vectors, meaning the scoring rule market (Interpretation 4) will be strictly less expressive than the others.
Somewhat conversely, consider $G$ to be negative Shannon entropy, which gives rise to the log scoring rule $S(\p,y) = \log p_y$.
Here the liquidity does become infinite on the boundary, and consequently the scoring rule vectors have infinite entries.
These vectors cannot be captured by any $\q\in\reals^n$, only in the limit.

For these two reasons, we restrict to $\relint\Delta_n$ in Lemma~\ref{lem:cost-scoring-duality}.
The first issue is somewhat surmountable, however: if one defines $\dG_\p$ to be the most favorable $\q$ (modulo $\ones$) tangent to $G$ at a boundary point $\p$, then all of the Pareto optimal trades will still be available to the scoring rule market trader.
Indeed, the only trades missing are those at the maximum possible price, which is Pareto-suboptimal for the trader anyway---consider a two outcome example when the price of the first security is 1, the maximum possible, so that purchasing the first security at this price is weakly worse than simply refraining from trade.
Thus, while in Theorem~\ref{thm:equivalence-interpretations} we assume there is a ``log-like'' LP, typically the market creator, which keeps the price away from the boundary, in principle one could generalize this statement using the ideas above to the case where liquidity runs out.

\begin{theorem}\label{thm:equivalence-interpretations}
  Let $C_i = G_i^*$ for generating functions $G_i$, where at least one $G_i$ is a pseudobarrier.
  Then the interpretations 1-5 above are equivalent in the following sense:
  given a coherent market state, the set of valid trades is identical, and the resulting market state is coherent.
\end{theorem}

Implicit in the proof of Theorem~\ref{thm:equivalence-interpretations} is that, in interpretations 1, 3, and 5, if the market state is not initially coherent, it becomes coherent after a sufficiently large trade.

\section{Defining liquidity as matrix-valued price insensitivity} \label{sec:liquidity_discussion}

Informally, liquidity is the extent to which assets/securities can be exchanged.
One way to capture liquidity, locally around a given price, is to quantify the extent to which the price stays stable during a transaction.
In other words, the lower the rate of change of the price, the higher the liquidity.

Thus far in the prediction market literature, even for large numbers of securities, liquidity is often captured by a single parameter.
The most popular approach is to consider some base cost function $C$ and define $C_\eta$ via the \emph{perspective transform} $C_\eta(\q) = \eta C(\q/\eta)$, where $\eta$ corresponds to the liquidity of the market~\citep{othman2013practical,li2013axiomatic,abernethy2014general,othman2011liquidity-sensitive}.
Another approach is to define the liquidity of $C$ to be some function of its Hessian $\nabla^2 C$ such as the inverse of its norm~\citep{abernethy2013efficient}.
A noted exception is \citet{dudik2014market}, where liquidity is acknowledged to depend on which specific subset of securities is under consideration.
Our approach is closest to the latter: as we argue soon in \S~\ref{sec:liquidity-hessian}, liquidity is indeed inherently multidimensional.
Any subspace of securities could have any degree of liquidity.
It is true that for very restricted protocols such as Uniswap V2, which corresponds to the perspective transform, liquidity has a fixed shape that can be scaled by a single real parameter.
In general, however, liquidity must be captured by a higher-dimensional object.
We begin by considering the case of 2 securities and constructing a scalar measure of liquidity that naturally corresponds to price insensitivity.
Then, we construct a matrix-valued measure of liquidity that recovers our scalar measure for the 2-security case.
We show that we can use our matrix-valued liquidity measure to recover a scalar liquidity measure \emph{in the direction} of any trade.

\subsection{Scalar-valued liquidity for 2 securities} \label{sec:liquidity-scalar}

Consider a prediction market with 2 securities.
We first note that by $\ones$-invariance, a 2-dimensional cost function may be specified by a 1-dimensional convex function.
Given any $\ones$-invariant cost function $C:\reals^n\to\reals$, we may write $C(\q) = c(q_1-q_2) + q_2$ for some convex $c:\reals\to\reals$ given by $c(q) = C((q,0))$.
Letting $\q = (q,0)$, then, the price of security 1 is $\nabla C(\q) = c'(q)$.
Note that we additionally use the lowercase variants of $C$ and $G$ in \S~\ref{sec:DeFi}.

Appealing to the above notion of liquidity as a function of the local price, let us define the liquidity at price $p$ to be the reciprocal of the rate of change of the price when the price is $p$.
While not required in general, for the purpose of this derivation, suppose that $c'' > 0$ everywhere.
Then we may define the liquidity at price $p\in[0,1]$ as $\ell(p) = 1/c''(q) > 0$, where $p = c'(q)$.
Since $\ell$ is strictly positive, we find that $\ell(p) = g''(p)$ for some convex function $g:[0,1]\to\reals$.
This relationship is in essence a special case of convex conjugate duality: we may simply take $g = c^*$.
From this duality, we have $g' = (c')^{-1}$, which is well-defined as $c'$ is strictly monotone; by the inverse function theorem, we could equivalently derive $\ell$ as $\ell(p) = ((c')^{-1})'(p) = g''(p)$.

\subsection{Liquidity as a Hessian matrix} \label{sec:liquidity-hessian}
In the 2-security case, we showed that liquidity can be measured as $\ell(p)=1/c''(q)=g''(p)$, adopting the 1-dimensional expression of the market.
In higher dimensions, we can analogously define the liquidity at price $\p$ to be $\ell(\p) = (\nabla^2 C(\q))^+$ at any vector $\q$ with price $\nabla C(\q) = \p$, where $A^+$ is the pseudoinverse of $A$.\footnote{The pseudoinverse is needed as the Hessian $\nabla^2 C$ is rank-deficient since $C$ is always flat in the $\ones$ direction.
  This observation also explains why we can express liquidity between 2 securities in one real number, since there is only 1 free parameter in $\nabla^2 C$ in that case.}
Here, liquidity is a matrix, which specifies the (inverse) rate of change of the price in any direction (or more generally, subspace) of interest.
Again appealing to convex duality, we can write $\q = \nabla \overline G(\p)$, where $\overline G$ is the 1-homogeneous extension of the dual function $G = C^*$.
Thus, we have $\ell(\p) = \left(\nabla^2 C(\nabla \overline G(\p))\right)^+$, or equivalently, $\ell(\p) = \nabla^2 \overline G(\p)$.
We may recover the usual measure of scalar liquidity as price insensitivity in the direction of any price $\hat{\p} \neq \p$.
Let $H(\p)=\nabla^2(\overline G(\p))$, and consider the product $(\hat{\p})^\top H(\p)\hat{\p}$.
This gives us the second derivative with respect to $t$ of the function $g(t)=f(\p+t\hat{\p})$, representing the price insensitivity in the direction of $\hat{\p}$ locally at $\p$.
Consider again the generating functions used in the example from \S~\ref{sec:example_run}:
\begin{align*}
& G_1(\p) = p_1\log(p_1)+p_2\log(p_2)-(p_1+p_2)\log(p_1+p_2), \\
& G_2(\p) = p_2\log(p_2) + p_3\log(p_3)-(p_2+p_3)\log(p_2+p_3).
\end{align*}
As $G_1$ and $G_2$ are already 1-homogeneous, they are equal to their 1-homogeneous extensions and we may take their Hessians directly.
Respectively, these are
\begin{align*}
H_1 &= \nabla ^2 G_1 = 
\begin{bmatrix}
\frac{1}{p_1} - \frac{1}{p_1+p_2} & - \frac{1}{p_1+p_2} & 0 \\
- \frac{1}{p_1+p_2} & \frac{1}{p_2} - \frac{1}{p_1+p_2} & 0 \\
0 & 0 & 0 \\
\end{bmatrix} = 
\begin{bmatrix}
\frac{1}{p_1}& 0 & 0 \\
0 & \frac{1}{p_2} & 0 \\
0 & 0 & 0 \\
\end{bmatrix} - 
\begin{bmatrix}
\frac{1}{1-p_3} & \frac{1}{1-p_3} & 0 \\
\frac{1}{1-p_3} & \frac{1}{1-p_3} & 0 \\
0 & 0 & 0 \\
\end{bmatrix},
\\
H_2 &= \nabla ^2 G_2 = 
\begin{bmatrix}
0 & 0 & 0 \\
0 & \frac{1}{p_2} - \frac{1}{p_2+p_3} & - \frac{1}{p_2+p_3}\\
0 & - \frac{1}{p_2+p_3} & \frac{1}{p_3} - \frac{1}{p_2+p_3}\\
\end{bmatrix} = 
\begin{bmatrix}
0 & 0 & 0 \\
0 & \frac{1}{p_2} & 0 \\
0 & 0 & \frac{1}{p_3} \\
\end{bmatrix} - 
\begin{bmatrix}
0 & 0 & 0 \\
0 & \frac{1}{1-p_1} & \frac{1}{1-p_1} \\
0 & \frac{1}{1-p_1} & \frac{1}{1-p_1} \\
\end{bmatrix}.
\end{align*}

Note that $G_1$ is flat in the direction of $\hat{\p}=(0,0,1)$, and $G_2$ is flat in the direction of $\hat{\p}=(1,0,0)$.
As we observed in our earlier example, neither function provides any liquidity to trades involving those assets.
But, by  adding the generating functions, and therefore adding their Hessians, we get
\begin{align*}
\nabla^2 G = 
\begin{bmatrix}
\frac{1}{p_1} - \frac{1}{p_1+p_2} & - \frac{1}{p_1+p_2} & 0 \\
- \frac{1}{p_1+p_2} & \frac{2}{p_2} - \frac{1}{p_1+p_2} - \frac{1}{p_2+p_3} & - \frac{1}{p_2+p_3} \\
0 & - \frac{1}{p_2+p_3} & \frac{1}{p_3} - \frac{1}{p_2+p_3}
\end{bmatrix},
\end{align*}
and liquidity is present for all assets.

\section{An impossibility result on the design of trading fees} \label{sec:fees}

Our protocol calls for assessing a cash fee to be paid by the trader and to each LP after each trade, parameterized by the trade vector, liability, and cost/generating function of the LP.
We do not inherently place any other restrictions on the fee function.

\subsection{Natural fees ``break'' when $n\geq 3$}

In practice, the design of fees has not been studied for markets with $n\geq 3$ securities.
Commonly used fees, like those used in Uniswap, exhibit deeply undesirable behavior when extended to these markets; see Appendix~\ref{sec:app-uniswap-fee}.
Here, we enumerate a set of axioms describing desirable behavior of the fee function.
Unfortunately, we show that it is impossible for a fee function to satisfy all properties.
We denote the fee that LP $i$ receives by $\fee_i(\{\vec{r}^i\},\{\vec{q}^i\})\in \reals$.
Let the fee that a trader pays be given by $\fee_{\mathrm{T}}(\{\vec{r}^i\},\{\vec{q}^i\})\in\reals$. 

\begin{axiom}[Budget Balanced, (BB)] \label{axi:bb}
    The sum of the fees collected by the LPs must be equal to the fee that a trader pays. That is, for all $\{\r^i\},\{\q^i\}$,
    \begin{equation*}
        \fee_{\mathrm{T}}(\{\vec{r}^i\},\{\vec{q}^i\})=\sum_{i=1}^k \fee_i(\{\vec{r}^i\},\{\vec{q}^i\}).
    \end{equation*}
\end{axiom}
\begin{axiom}[Trader Simple, (TS)] \label{axi:ts}
    The fee paid by a trader should be independent of how the trade is split across LPs and their individual liquidities and should depend only on the aggregate trade $\r$ and aggregate liquidity $\q$.
    The axiom captures the goal of keeping the trader interface simple and abstracted from the inner dynamics of LPs.
    Hence, we require that there exists \(\overline{\fee}_{\mathrm{T}}\) such that for all \(\{\r^i\},\{\q^i\}\),
    \begin{equation*}
        \fee_{\mathrm{T}}(\{\vec{r}^i\},\{\vec{q}^i\})=\overline{\fee}_{\mathrm{T}}(\vec{r},\vec{q})
    \end{equation*}
    whenever $\vec{r}=\sum_i\vec{r}^i$ and $\vec{q}=\sum_i\vec{q}^i$.
\end{axiom}
\begin{axiom}[LP Decomposability, (LD)] \label{axi:ld}
    The fee LP $i$ charges should only depend on inputs $\r^i,\q^i$ to enable the parallel market interpretation proposed in \S~\ref{sec:interpretations}.
    Hence there exists $\fee_{\mathrm{LP}}$ such that for all $i\in\{1,\ldots,k\}$,
    \begin{equation*}
        \fee_i(\{\vec{r}^i\},\{\vec{q}^i\})=\fee_{\mathrm{LP}}(\vec{r}^i,\vec{q}^i).
    \end{equation*}
\end{axiom}

\begin{axiom}[Nonnegativity, (NN)] \label{axi:nn} 
    All fees must be nonnegative, so for all $\{\r^i\},\{\q^i\}$,
    \begin{align*}
        \fee_i(\{\r^i\},\{\q^i\}),\fee_{\textsc{T}}(\{\r^i\},\{\q^i\})&\in\mathbb{R}^n_{\geq0}
    \end{align*}
    and $\fee_i(\{\r^i\},\{\q^i\}), \fee_{\textsc{T}}(\{\r^i\},\{\q^i\}) = 0$ if and only if $\r^i=\0$ or $\r=\0$ respectively.
\end{axiom}

The above axioms are sufficient for our impossibility result; see \S~\ref{sec:fee-practice} for an additional Axiom~\ref{axi:pf}.
We start by stating a useful lemma.
See Appendix~\ref{sec:fees-appendix} for proofs.
\begin{lemma} \label{lem:feeuni}
    Axioms~\ref{axi:bb}, \ref{axi:ts}, and \ref{axi:ld} allow us to write 
    $
        \overline{\fee}_\mathrm{T}=\fee_\mathrm{LP}.
    $
\end{lemma}
Due to Lemma~\ref{lem:feeuni}, we use \(\fee()\) from now on in the place of functions \(\overline{\fee}_\mathrm{T}()\) and \(\fee_\mathrm{LP}()\).
We now state our impossibility result, providing a proof sketch but deferring the full constructive proof to the appendix.
\begin{theorem} \label{thm:fee-impossible}

Axioms~\ref{axi:bb}, \ref{axi:ts}, \ref{axi:ld}, and \ref{axi:nn} are incompatible.
\end{theorem}
\begin{proof}[Proof Sketch]
We approach the problem from the market scoring rule point of view, considering a 3 asset instance.
We consider the three generating functions $G_1=-2\sqrt{2p_2p_3},G_2=-2\sqrt{2p_1p_3}$, and $G_3=-2\sqrt{2p_1p_2}$, which are symmetric (i.e., equal with each other up to permutation of assets), 1-homogeneous, and flat with respect to one asset.
We consider a setting in which LPs with generating functions $G_1$ and $G_2$ are present. We set initial and final prices such that the trades facilitated are $\r^1=(0,-1,1)$ and $\r^2=(1,0,-1)$, and therefore the overall trade is $\r^1+\r^2=(1,-1,0)$.
We then consider a setting in which only LP$_3$ with $G_3$ is present, and set prices to create a direct trade of $\r=\r^3=(1,-1,0)$.
We add ``filler'' LPs to both settings that add liquidity so that the starting liability vector is the same in both settings, but are flat with respect to both asset 1 and asset 2, and therefore do not participate in the trade. By Axiom~\ref{axi:nn}, the fees that these ``filler'' LPs collect is $\0$, since the trade facilitated by these LPs is $\0$.

In both settings, the trader makes a trade of $(1,-1,0)$, with the same starting liability $\q$, and so pays the same fee.
By Lemma~\ref{lem:feeuni}, we have
\begin{align*}
\bar{\fee}_T(\r,\q) = \fee_1(\r^1,\q)+\fee_2(\r^2,\q),\quad \bar{\fee}_T(\r,\q)= \fee_3(\r^3,\q).
\end{align*}

Now, we consider an alternative pair of scenarios, in which we permute the assets and the roles of the LPs.
With LP$_1$ and LP$_3$ present, we set prices so that $\r^1=(0,1,-1)$, the negative of its original value, $\r^3=(1,-1,0)$, and the total trade is $(1,0,-1)$. In the second scenario, LP$_2$ facilitates a trade of $\r^2=(1,0,-1)$.
Note that both $\r^2$ and $\r^3$ have the same value throughout.
 We again introduce ``filler'' LPs so that the overall initial liability is the same as in the previous pair of settings.
We then have
\begin{align*}
\bar{\fee}_T(\r,\q) = \fee_1(-\r^1,\q)+\fee_3(\r^3,\q),\quad \bar{\fee}_T(\r,\q) = \fee_2(\r^2,\q).
\end{align*}
So we have that $\fee_1(\r^1,\q)+\fee_2(\r^2,\q)=\fee_3(\r^3,\q)$ and $\fee_1(-\r^1,\q)+\fee_3(\r^3,\q)=\fee_2(\r^2,\q).$ But by Axiom~\ref{axi:nn}, all these fees are strictly positive, and so this is impossible.
\end{proof}

\begin{figure}[htbp]
    \centering
    
    \begin{subfigure}[b]{0.45\textwidth}
        \centering
        \begin{tikzpicture}[scale=0.5,
    >={Stealth[round, length=3mm]},
    thick,                         
    main node/.style={font=\Large, outer sep=3pt},
    edge label/.style={font=\small, auto}
]

    \node[main node] (1) at (90:3) {1};
    \node[main node] (2) at (210:3) {2};
    \node[main node] (3) at (330:3) {3};

    \draw[->] (3) -- node[edge label] {LP3} node[edge label, swap] {$r^2$} (1);
    \draw[->] (2) -- node[edge label] {LP1} node[edge label, swap] {$r^1$} (3);
    \draw[->, dotted, very thick] (2) -- node[edge label] {$r^3=r^1+r^2$} (1);

\end{tikzpicture}
        \caption{trades $r^1$ and $r^2$ facilitated by LP1 and LP2 add to $r^3$.}
        \label{fig:A}
    \end{subfigure}
    \hfill
    \begin{subfigure}[b]{0.45\textwidth}
        \centering
        \begin{tikzpicture}[scale=0.5,
    >={Stealth[round, length=3mm]},
    thick,                         
    main node/.style={font=\Large, outer sep=3pt},
    edge label/.style={font=\small, auto}
]

    \node[main node] (1) at (90:3) {1};
    \node[main node] (2) at (210:3) {2};
    \node[main node] (3) at (330:3) {3};

    \draw[->] (2) -- node[edge label] {$r^3$} node[edge label, swap] {LP3} (1);

\end{tikzpicture}
        \caption{Trade $r^3$ facilitated directly by LP3.}
        \label{fig:B}
    \end{subfigure}
    
    \vspace{1cm}
    
    \begin{subfigure}[b]{0.45\textwidth}
        \centering
        \begin{tikzpicture}[scale=0.5,
    >={Stealth[round, length=3mm]},
    thick,                         
    main node/.style={font=\Large, outer sep=3pt},
    edge label/.style={font=\small, auto}
]

    \node[main node] (1) at (90:3) {1};
    \node[main node] (2) at (210:3) {2};
    \node[main node] (3) at (330:3) {3};

    \draw[->] (2) -- node[edge label] {LP3} node[edge label, swap] {$r^3$} (1);
    \draw[->] (3) -- node[edge label] {LP1} node[edge label, swap] {$r^1$} (2);
    \draw[->, dotted, very thick] (3) -- node[edge label, right] {$r^2=r^1+r^3$} (1);

\end{tikzpicture}
        \caption{Trades $r^1$ and $r^3$ facilitated by LP1 and LP3 add to $r^2$.}
        \label{fig:C}
    \end{subfigure}
    \hfill
    \begin{subfigure}[b]{0.45\textwidth}
        \centering
        \begin{tikzpicture}[scale=0.5,
    >={Stealth[round, length=3mm]},
    thick,                         
    main node/.style={font=\Large, outer sep=3pt},
    edge label/.style={font=\small, auto}
]

    \node[main node] (1) at (90:3) {1};
    \node[main node] (2) at (210:3) {2};
    \node[main node] (3) at (330:3) {3};

    \draw[->] (3) -- node[edge label] {LP$2$} node[edge label, swap] {$r^2$} (1);

\end{tikzpicture}
        \caption{Trade $r^2$ facilitated directly by LP2}
        \label{fig:D}
    \end{subfigure}
    
    \caption{A set of trades that, with the starting liquidity set to be equal in all cases, results in
    $\fee_3=\fee_1+\fee_2$ and
    $\fee_2=\fee_1+\fee_3$,
    violating Axiom \ref{axi:nn}.
    }
    \label{fig:grid}
\end{figure}
\vspace{-17pt}
\subsection{Fee design in practice and future work} \label{sec:fee-practice}

In light of our impossibility result, what fee structures are best in practice?
Many markets use $\fee_i(\r^i,\q^i)=\beta\|\r^i\|_1$ or $\fee_i(\r,\q) = \beta\|\r\|\frac{\|\r^i\|}{\sum_j\|\r^j\|}$ for some $0<\beta<1$, or some minor variant thereof. Uniswap uses $\fee_i(\r^i,\q^i)=\beta(-\r^i)_+$.
We provide a worked example of Uniswap's fee in Appendix~\ref{sec:app-uniswap-fee}.
While these fees violate trader simplicity (Axiom~\ref{axi:ts}), any violation of the sort we construct in the proof of Theorem~\ref{thm:fee-impossible} is ``reasonable'' in the sense that the trader must pay a higher fee for an ``indirect'' trade that requires the participation of multiple LPs. Axiom~\ref{axi:pf}, given below, is also violated, but only when the difference between the trader's belief and the price for any security they are trading is less than $\beta$.

A possible approach is to allow fees that are a function not only of $\r$ and $\q$ but also the cost function/scoring rule set by the liquidity provider.
Surprisingly, adding this additional parameter allows us to overcome the impossibility result of Theorem~\ref{thm:fee-impossible} and satisfy Axioms 1 through 4.
However, while such fees do not violate Axiom~\ref{axi:ts} as written, they violate the spirit of the axiom, as the trader fee depends on the configuration of LPs.
Despite this concern, such fees exist whose dependence on $G$ can be reduced to trader-interpretable quantities such as $\p$ and $\hat{\p}$.
One such fee is $\r\cdot(\hat{\p}-\p)$, which we use to prove that a fee function parameterized by $G$ may satisfy Axioms 1 through 4.
\begin{proposition}
A fee function of the form $\fee(\r,\q,G)$ may satisfy Axioms 1 through 4.
\end{proposition}
\begin{proof}
To show that Axioms~\ref{axi:bb}, \ref{axi:ts} and \ref{axi:ld} hold, first note that for any fixed $\q$, $\p$, legal trade $\r$, and set of LPs with generating functions $\{G_i\}$, the final price $\hat{\p}$ after performing trade $\r$ is exactly determined by $\q,\p$, and $\{G_i\}$. We then observe that
\[
\fee(\r,\q)=\r\cdot(\hat{\p}-\p)=\sum_{i=1}^k \r^i\cdot(\hat{\p}-\p)=\sum_{i=1}^k \fee_i(\r^i,\q^i).
\]
To show that Axiom~\ref{axi:nn} is satisfied, note that whenever an element $r_j$ of $\r$ is positive, so is $\hat{p}_j-p_j$, and whenever $r_j$ is negative, so is $\hat{p}_j-p_j$, and so $\r\cdot(\hat{\p}-\p)=\sum_{j=1}^n r_j(\hat{p}_j-p_j) \geq 0$, with equality exactly when $\r=\0$.
\end{proof}
\begin{axiom}[Profitability, (PF)] \label{axi:pf}

If $\r$ is a legal trade given $\q$ and $C$, then $\r^* := \sum_i \max\{r_i,0\}$ is the maximum possible profit that a trader can make off of $\r$,
if all acquired securities are sold at price 1.
Then $\fee_{\textsc{T}}(\{\r\},\{\q\}) \leq \r^*$.
\end{axiom}
However, this fee still exhibits very undesirable behavior.
A trader may reduce their fee by dividing their trade into smaller sub-trades, each of which creates a smaller difference in price.
Indeed, for sufficiently large trades, the fee can be arbitrarily large, even though the maximum profit of any trade is bounded by the deposited liquidity, severely violating Axiom~\ref{axi:pf}.
For example, given a pseudobarrier generating function $G$ on three assets, any trade $\r$ to $\hat{\p}=(1-\epsilon,0,0)$ will be bounded in the first term by the maximum deposit of $G$ but approach $-\infty$ in the other two terms.
Therefore, for any starting price $\p$, the fee for trading to $(1-\epsilon,0,0)$ grows arbitrarily large as $\epsilon\to 0$, while the trader's maximum profit is bounded above by $-\min_{\q\in\Delta_3}G(\q)$.
Ultimately, these concerns nullify any practical use for the fee.
The design of simple, liquidity-dependent fees that are more practical for future use is a promising direction of future work.

\section{Recovering and extending protocols from decentralized finance}\label{sec:DeFi}
In this section, we discuss our contributions to the decentralized finance literature.
We recover several common protocols, like Uniswap V2 and V3, as special instances of our general protocol.
For this section, we restrict to the exchange of two securities, as is common in decentralized finance.
In Appendix~\ref{sec:app-general-two-asset}, we provide a general protocol (Protocol~\ref{alg:general-two-asset-g}) for this setting.
In Appendix~\ref{sec:new_protocols_detailed}, we use the flexibility of our general protocols to propose new ones, to be used either in decentralized finance or for prediction markets.

\subsection{Conventional differences}\label{subsec:conventions}

There are many conventional differences to note between 
the prediction market and the decentralized finance AMM literature. 
For example, with regard to prices, decentralized finance typically uses an ``exchange-rate'' version of the contract price: the rate at which one can exchange one asset for another.
Taking advantage of the structure of cost functions, we instead adopt a \textit{normalized price} convention.
One can view normalized prices $\p \in \Delta_n$ as an exchange rate between assets and the ``grand bundle'' $\ones$ of all assets; $p_i$ denotes the instantaneous price, in units of $\ones$, to purchase asset $i$.

Converting between the two conventions is straightforward.
Given normalized prices $\p$, one can simply define the exchange rate between $i$ and $j$ as $\hat p_{ij} = p_i / p_j$.
Conversely, given pairwise exchange rates, one can define $\x = (1, \hat p_{21}, \hat p_{31}, \ldots, \hat p_{n1})$ and take $\p = \x / \|\x\|_1$.
The conversion simplifies in the case of two assets, as $\hat p_{12}=\frac{p}{1-p}$ and $p = \hat p_{12}/(\hat p_{12}+1)$.

In the decentralized finance literature, trades are also typically oriented toward the market maker, whereas we instead consider all trades to be oriented towards the trader.
Also, CFMMs in decentralized finance track the reserves held by the market maker, while cost function prediction markets typically track their liabilities as a function of the eventual outcome.
So, a vector $\res$ of reserves corresponds to a vector $-\q$ of liabilities.

\subsection{Uniswap V2}
\label{sec:uniswap-v2}

Uniswap V2, introduced by~\citet{adams2020whitepaperv2}, is a commonly used AMM in the Ethereum ecosystem to trade assets and has the functional invariant $\varphi_\alpha(\res)=\xa\xb=\alpha^2$.
We track the reserve vector $\res$, so that $\xa$ and $\xb$ represent the amount of assets 1 and 2, respectively, held by the market maker.
We also still use normalized prices, leading to minor differences from the literature.
The normalized price of the first asset can be computed as $p=\frac{\xb}{\xa+\xb}$, while the exchange rate model has $\hat{p}=\frac{x_2}{x_1}$.
We refer the reader to~\citet{fan2022differentialliquidity, fan2023strategicliquidity} for a detailed breakdown of Uniswap V2 mechanics.

Uniswap V2 restricts how an LP can add or remove their liquidity by constraining them to use the same base function $g_0=-2\sqrt{p(1-p)}$ and to express their liquidity only via a parameter $\alpha^i$ where $g_i=\alpha^ig_0$.
We state the Uniswap V2 protocol as Protocol~\ref{alg:V2-two-asset} explicitly in Appendix~\ref{sec:app-uniswapv2}.
Recall that from \citet{fan2022differentialliquidity}, the bundle required to change liquidity from $\alpha^i$ to $\alpha'$ while keeping the price invariant is $\left(\frac{\alpha'-\alpha^i}{\sqrt{\hat{p}}},(\alpha'-\alpha^i)\sqrt{\hat{p}}\right)$. Using normalized prices, this is represented as $\left((\alpha'-\alpha^i)\sqrt{\frac{1-p}{p}},(\alpha'-\alpha^i)\sqrt{\frac{p}{1-p}}\right)$ in our protocol.
We note that instead of skimming $\gamma$ from $(-\r)_+$ for trading fees, we ask for $\beta(-\r)_+$ from the trader when they request the trade $\r$.
These two fee schemes are equivalent when $\beta=\frac{\gamma}{1+\gamma}$.
\begin{lemma}[Informal]\label{lemma:uniswap_is_subcase}
    Protocol \ref{alg:V2-two-asset} is a special case of Protocol \ref{alg:general-two-asset-g} for specific restrictions on $\Ginit$.
\end{lemma}

We defer the formal statements and proofs to Appendix~\ref{sec:omitted_proofs-defi}.
There, we show that the liability vectors from the latter indeed satisfy the constant product invariant for our choice of generating functions using Propositions~\ref{prop:is_uniswap} and \ref{prop:protocol_equal}.

\subsection{Uniswap V3 and general bucketing}\label{sec:general_bucketing}

Uniswap V2 requires LPs to provide liquidity on the entire price space. This restriction may look intuitive, but it is suboptimal since liquidity allocations far from the current price may not be used.
For example, a market that trades securities on a sports game might not benefit from having liquidity at prices in the $(0,0.1)$ range, say.
Moreover, when LPs provide liquidity, they take on the risk of price volatility, and ideally we would like to allow them to bound that risk.
On the other hand, it is computationally challenging to maintain an infinitely flexible LP protocol.

Motivated by these concerns, the Uniswap V3 protocol, \citep{adams2021whitepaperv3}, partitions the price space, allowing each LP $i$ to contribute a proportion $\alpha^{ij}$ of liquidity on any price bucket $[a_j,b_j]$ of their choosing.
We refer the reader to \citet{fan2022differentialliquidity,fan2023strategicliquidity} for a detailed analysis of Uniswap V3.
In Appendix~\ref{sec:uniswapV3} we show that Uniswap V3, given by Protocol~\ref{alg:V3-two-asset}, is a special case of our general protocol.

In this section, we \emph{generalize} the idea of bucketing with regards to an arbitrary cost function $C$.
For ease of exposition, let us restrict again to the two-outcome setting.
We do so as this allows for LPs to be more expressive, depending on the number of price intervals and also keeps the complexity of implementation from blowing up.
Let $G(\p)$ indicate the corresponding dual where $p_2=1-p_1$. 
Let $\ell^{(j)}$ be its liquidity function restricted to the $j$-th price interval $ [a_j,b_j]$ of $p_1$ and be given by 
$
  \ell^{(j)}=\nabla^2\bar{G}\ones_{p_1\in[a_j,b_j]}.
$
Let the liability vector associated with the corresponding cost function dual for $\ell^{(j)}$ be given below, which we derive in Appendix~\ref{sec:bucket_liquidity_vec}.
\begin{align*}
  \liabilityf(C^{(j)}) 
  &= \begin{pmatrix}
      S_g(\max\{a_j,p_1\},1)-S_g(\max\{b_j,p_1\},1)
      \\
      S_g(\min\{b_j,p_1\},0)-S_g(\min\{a_j,p_1\},0)
  \end{pmatrix}
\end{align*}
where $\p=(p_1,p_2)$ is the current price, $S_g(p,y)=g(p)+g'(p)(y-p)$ and $G(\p)=g(p_1)$ as discussed in \S~\ref{sec:scoring-rules}, \S~\ref{sec:liquidity_discussion} and $C^{(j)}=(\iint \ell^{(j)})^*$.

In Table~\ref{tab:buckets}, we use the above general liability vector to state the liability vectors for Uniswap V3, as well as new bucketing protocols where we use $G(\p)=p_1\log p_1+p_2\log p_2$ from LMSR and $S_g(p,y) = -(p-y)^2$ of the Brier scoring rule as the base ``shapes.''
We defer detailed workings for the LMSR bucketing protocol to Appendix~\ref{sec:bucket_lmsr}.
\begin{table}[ht] 
  \begin{center} 
  \begin{tabular}{c|ccc}
    \hline
    & Uniswap V3 & LMSR & Brier \\
    \hline
    $p<a_j$ & $\alpha_j\begin{pmatrix}
      \sqrt{\frac{1-b_j}{b_j}}-\sqrt{\frac{1-a_j}{a_j}}\\0
    \end{pmatrix}$ & $\begin{pmatrix} 
      \log\frac{a_j}{b_j}
      \\
      0
    \end{pmatrix}$ & $\begin{pmatrix} 
      (1-b_j)^2-(1-a_j)^2
      \\
      0
    \end{pmatrix}$ \\
    $p\in[a_j,b_j]$ & $\alpha_j\begin{pmatrix}
      \sqrt{\frac{1-b_j}{b_j}}-\sqrt{\frac{1-p}{p}}\\\sqrt{\frac{a_j}{1-a_j}}-\sqrt{\frac{p}{1-p}}
    \end{pmatrix}$& $\begin{pmatrix}
      \log\frac{p}{b_j} \\
      \log\frac{1-p}{1-a_j}
    \end{pmatrix}$& $\begin{pmatrix}
      (1-b_j)^2-(1-p)^2 \\
      a_j^2-p^2
    \end{pmatrix}$\\
    $p>b_j$ & $\alpha_j\begin{pmatrix}
      0\\\sqrt{\frac{a_j}{1-a_j}}-\sqrt{\frac{b_j}{1-b_j}}
    \end{pmatrix}$& $\begin{pmatrix}
      0 \\ \log\frac{1-b_j}{1-a_j}
    \end{pmatrix}$
    & $\begin{pmatrix}
      0 \\ a_j^2-b_j^2
    \end{pmatrix}$\\
    \hline
  \end{tabular}
\end{center}
  \caption{Liability vectors for bucketing liquidity protocols.\label{tab:buckets}} 
  \vspace{-23pt}
\end{table}

\section{Discussion and open directions} \label{sec:discussion-open-dir}

We have given a general protocol for liquidity provisioning in prediction markets (\S~\ref{sec:protocol-pred-markets}), and more broadly any automated market making setting including decentralized finance.
In a sense that we formalize in \S~\ref{sec:equivalence} and \S~\ref{sec:liquidity_discussion}, this protocol is maximally expressive, though it can be restricted for computational convenience, or LP ease-of-use.
We also provide an impossibility result on the design of trading fees (\S~\ref{sec:fees}).
One avenue for future work is the design of simple liquidity-dependent fees that can work in practice.
Another direction is to further explore specific restrictions, and study the tradeoffs in expressiveness and computation, and the implications for market efficiency.

When implementing Protocol~\ref{alg:general-n-asset-cost} in practice, there is a tradeoff between LPs' expressiveness and the computational cost of running the protocol. 
This tradeoff is a strong consideration in decentralized finance, as all computations must be done on-chain. 
As prediction markets are typically not as constrained, much more expressive protocols are possible.
For example, one could allow LPs to specify their liquidity functions via polynomials of bounded degree, or weighted sums of basis functions, or any computationally convenient class of functions that well approximate any possible liquidity allocation.
In decentralized finance, Protocol~\ref{alg:general-n-asset-dual-price} seems to be the more practically appealing of the two.
Instead of specifying a trade directly, the ``inversion'' from proposed trades to implied prices is handled off-chain by the trader.
The remaining computation, of the implied liability vectors, is more straightforward.
We leave the analysis of further restrictions and tradeoffs to future work.

\pagebreak
\bibliographystyle{plainnat}
\bibliography{refs,cfmms}
\pagebreak
\appendix

\section{Proofs From Section~\ref{sec:equivalence}}
\label{sec:app-proofs-of-equivalence}

We prove in Appendix~\ref{sec:equiv-full-prot} that Protocol~\ref{alg:general-n-asset-dual-price} and Protocol~\ref{alg:general-n-asset-cost} are essentially equivalent under some mild conditions.

\subsection{Technical lemmas}

We begin with some standard facts from convex analysis.
\begin{proposition}[{\citet[Theorem 23.5]{rockafellar1997convex}}]\label{prop:convex-duality}
  Let $f:\reals^n\to\reals\cup\{\infty\}$ be a closed convex function and $f^*$ its conjugate.
  Then for all $\x,\x^*\in\reals^n$ the following are equivalent:
  \begin{enumerate}[itemsep=0pt]
  \item $\x^* \in \partial f(\x)$
  \item $\x \in \partial f^*(\x^*)$
  \item $f(\x) + f^*(\x^*) = \inprod{\x^*}{\x}$
  \end{enumerate}
\end{proposition}

\begin{proposition}\label{prop:infimal-conv}
  Let $f_i:\reals^n\to\reals\cup\{\infty\}$ be convex for $i\in\{1,\ldots,k\}$
  such that $\bigcap_i \relint\dom f_i \neq \emptyset$.
  Let $f = \sum_i f_i$.
  Then $f^* = \bigwedge_i f_i^*$, and the infimum in $\bigwedge$ in the definition of $f^*$ is always attained.
  Moreover, for $\v \in \relint\dom f^*$, and any split $\v = \sum_i \v^i$, we have
  $f^*(\v) = \sum_i f_i^*(\v^i)$ if and only if $\bigcap_i \partial f_i(\v^i) \neq \emptyset$.
\end{proposition}
\begin{proof}
  The first statement follows from \citet[Theorem 16.4]{rockafellar1997convex}; we will prove the second.
  First suppose $f^*(\v) = \sum_i f_i^*(\v^i)$ for $\v\in\relint\dom f^*$ and $\v = \sum_i \v^i$.
  From \citet[Theorem 23.4]{rockafellar1997convex}, $\partial f^*(\v) \neq \emptyset$.
  Now \citet[Theorem 3.6]{stromberg1994study} gives $\bigcap_i \partial f_i^*(\v^i) = \partial f^*(\v) \neq \emptyset$.  \footnote{As needed for that result to apply, the assumption that there exists some $\x \in \bigcap_i \relint\dom f_i$ implies that $f^*_i$ is bounded from below by the same affine function, namely one with gradient $\x$.}

  For the converse, let $\x \in \bigcap_i \partial f_i^*(\v^i)$ and define $\v = \sum_i \v^i$.
  From Proposition~\ref{prop:convex-duality}, $\v^i \in \partial f_i(\x)$ for all $i$.
  Now \citet[Theorem 23.8]{rockafellar1997convex} gives $\v = \sum_i \v^i \in \partial f(\x)$.
  Proposition~\ref{prop:convex-duality} again implies
  gives $f^*(\v) = \inprod{\x}{\v} - f(\x) = \sum_i \inprod{\x}{\v^i} - \sum_i f_i(\x) = \sum_i f^*_i(\v^i)$.
\end{proof}

We now prove several technical lemmas we use in the proof of Theorem~\ref{thm:equivalence-interpretations}.
  
\begin{lemma}\label{lem:ndt-increasing}
  Let $G$ be a pseudobarrier and $C = G^*$.
  Then for all $\q,\hat\q\in\reals^n$ we have $\partial C(\q) \subseteq \relint\Delta_n$  and $\hat\q \nsucceq	 \q \implies C(\hat\q) > C(\q)$.
\end{lemma}
\begin{proof}
  Let $\p \in \partial C(\q)$.
  By Proposition~\ref{prop:convex-duality}, $\q \in \partial G(\p)$.
  If $\p \notin \relint\Delta_n$, then we have an interior sequence $\{\p^j\}_j$ such that $\p^j \to_j \p$ and subgradients $\q^j \to_j \q$, violating the definition of a pseudobarrier.
  For this $\p$, we have $p_y > 0$ for all $y \in \{1,\ldots,n\}$.
  As $\hat\q \nsucceq \q$, we have $\hat q_y \geq q_y$ for all $y$, with at least one inequality strict.
  Thus by the subgradient inequality, $C(\hat\q) - C(\q) \geq \inprod{\p}{\hat\q - \q} > 0$, as desired.  
\end{proof}

The following lemma is essentially a restatement of results due to \citet{ovcharov2018proper,ovcharov2015existenc}.
It says that subgradients of generating functions $G$ are unique modulo $\ones$.
\begin{lemma}\label{lem:smooth-G-unique-subgradients-mod-ones}
  Let $G$ be a generating function.
  Then for all $p\in\Delta_n$, and all $\q,\hat{\q} \in \partial G(\p)$, there exists $\alpha \in \reals$ such that $\hat{\q} = \q + \alpha\ones$.
\end{lemma}
\begin{proof}
  Let $\q \in \partial G(\p)$.
  We will show $\hat\q \in \partial \overline G(\p)$, where $\hat\q = \q + (G(\p) - \inprod{\q}{\p})\ones$.
  By construction, $\inprod{\hat\q}{\p} = \inprod{\q}{\p} + G(\p) - \inprod{\q}{\p} = G(\p)$.
  For all $\x\in\reals^n_{\geq 0} \setminus \{\0\}$, we have
  \begin{align*}
    \overline G(\x)
    &= \|\x\|_1 G(\x/\|\x\|_1)
    \\
    &\geq \|\x\|_1 \left(G(\p) + \inprod{\q}{\x/\|\x\|_1 - \p}\right)
    \\
    &= \|\x\|_1 \left(G(\p) + \inprod{\hat\q}{\x/\|\x\|_1 - \p}\right)
    \\
    &= \|\x\|_1 \left(\inprod{\hat\q}{\x/\|\x\|_1} + G(\p) - \inprod{\hat\q}{\p}\right)
    \\
    &= \inprod{\hat\q}{\x}
    \\
    &= \overline G(\p) - \inprod{\hat\q}{\p} + \inprod{\hat\q}{\x}
    \\
    &= \overline G(\p) + \inprod{\hat\q}{\x - \p}~.
  \end{align*}
  Thus, every element of $\partial G(\p)$, up to a shift by $\ones$, is an element of $\partial \overline G(\p)$.
  As the latter is a singleton set, by assumption on $G$, the result follows.
\end{proof}

\begin{lemma}
  \label{lem:infimal-conv-consistent-price}
  Let $C = \bigwedge_i C_i$ where $C^*_i$ are generating functions.
  
  Then for all $\q\in\reals^n$, the infimum in the definition of $C(\q)$ is attained, and $C(\q) = \sum_i C_i(\q^i)$ if and only there exists $\p\in\Delta_n$ such that $p \in \partial C_i(\q^i)$ for all $i$.
\end{lemma}
\begin{proof}
  We have $\dom C_i^* = \Delta_n$ and $\dom C_i = \reals^n$ for all $i$, so Proposition~\ref{prop:infimal-conv} applies.
  
\end{proof}

\begin{lemma}
  \label{lem:cost-scoring-duality}
  Let $C = G^*$ where $G$ is a generating function.
  For $\p\in\relint\Delta_n$, let $S(\p,y) = G(\p) - \inprod{\dG_\p}{\indvec{y} - \p}$, where $\{\dG_\p\in\partial G(\p) \mid \p\in\relint\Delta_n\}$ is a selection of subgradients.
  
  Then for all $\p\in\relint\Delta_n$ and $\q\in\reals^n$, we have $(\p \in \partial C(\q) \wedge C(\q) = 0) \iff \q = S(\p,\cdot)$.
  In particular,
  $\{S(\p,\cdot) \mid \p\in\relint\Delta_n\} = \{\q \in C^{-1}(0) \mid \partial C(\q) \cap \relint\Delta_n \neq \emptyset\}$.

\end{lemma}
\begin{proof}
  Let $\q = S(\p,\cdot)$.
  Then $C(\q) = C(\dG_\p + (G(\p) - \inprod{\dG_\p}{\p})\ones) = G(\dG_p) + G(\p) - \inprod{\dG_\p}{\p} = 0$
  by Proposition~\ref{prop:convex-duality}.
  Furthermore, by the same theorem, $\dG_\p \in \partial G(\p) \iff \p \in \partial C(\dG_\p)$, and by $\ones$-invariance, $\partial C(\dG_\p) = \partial C(S(\p,\cdot))$.
  Thus, $\p \in \partial C(\q)$.

  Now let $\q$ such that $C(\q) = 0$, and take $\p \in \partial C(\q)$; we will show $\q = S(\p,\cdot)$.
  By Proposition~\ref{prop:convex-duality}, $\q \in \partial G(\p)$.
  From Lemma~\ref{lem:smooth-G-unique-subgradients-mod-ones}, we have $\partial G(\p) = \{\dG_\p + \alpha\ones \mid \alpha \in\reals\}$.
  Thus $\q = \dG_\p + \alpha\ones$ for some $\alpha\in\reals$.
  Now $0 = C(\q) = C(\dG_\p + \alpha\ones) = C(\dG_\p) + \alpha = \inprod{\dG_\p}{\p} - G(\p) + \alpha$ by Proposition~\ref{prop:convex-duality}.
  Thus $\alpha = G(\p) - \inprod{\dG_\p}{\p}$, and we have $\q = \dG_\p + (G(\p) - \inprod{\dG_\p}{\p})\ones = S(\p,\cdot)$, as desired.
  
\end{proof}

\subsection{Proof of Theorem~\ref{thm:equivalence-interpretations}}
      
\begin{proof}
  We will show that 1 and 5 are each equivalent to 2, 3 is equivalent to 1, and 4 is equivalent to 1, and 3 is equivalent to 4.
  For each, let $\{\q^i\}_i$ be the current market state, $\q = \sum_i \q^i$, $\p$ a consistent price, and $C = \bigwedge_i C_i$.
  From Lemma~\ref{lem:infimal-conv-consistent-price}, price consistency implies
  $C(\q) = \sum_i C_i(\q^i)$, i.e., the $\q^i$ vectors achieve the infimum in the definition of the infimal convolution.
  \begin{enumerate}\setlength{\itemsep}{4pt}
  \item[($1\leftrightarrow 2$)]
    A trade for 2 satisfies the conditions for 1, so we only need to show that this choice is Pareto optimal for the trader; coherence will then follow by Lemma~\ref{lem:infimal-conv-consistent-price}.
    More formally, let $\{\r^i\}_i$ satisfy $C_i(\q^i + \r^i) = C_i(\q^i)$ for all $i$.
    Letting $\r = \sum_i\r^i$, from the definition of infimal convolution, $C(\q + \r) \leq \sum_i C_i(\q^i + \r^i) = \sum_i C_i(\q^i) = C(\q)$.
    We wish to show that $\r$ is Pareto optimal if and only if $C(\q + \r) = C(\q)$, or equivalently, that $\r$ is Pareto suboptimal if and only if $C(\q + \r) < C(\q)$.

    Suppose first that we had some $\hat \r \nsucceq \r$ such that $C_i(\q^i + \hat\r^i) = C_i(\q^i)$ where $\hat\r = \sum_i\hat\r^i$.
    From the same argument as above, we have $C(\q + \hat\r) \leq \sum_i C_i(\q^i + \hat\r^i) = \sum_i C_i(\q^i) = C(\q)$.
    By Lemma~\ref{lem:ndt-increasing}, $C(\q + \hat\r) > C(\q + \r)$, giving $C(\q + \r) < C(\q)$.
    Conversely, suppose $C(\q + \r) \neq C(\q)$, which from the inequality above implies $C(\q + \r) < C(\q)$.
    Let $\hat\r = \r + (C(\q) - C(\q+\r))\ones \nsucceq \r$.
    By $\ones$ invariance, we have $C(\q + \hat \r) = C(\q)$.
    From part (3) of the proof below, there exists a split $\hat \r = \sum_i \hat \r^i$ such that of $C_i(\q^i + \hat \r^i) = C_i(\q^i)$.
    Thus, $\r$ was not a Pareto-optimal total trade.
  \item[($2\leftrightarrow 5$)]
    Let $\r$ be a trade from interpretation 2, so that $C(\q + \r) = C(\q)$ where $C = \bigwedge_i C_i$.
    Set $\v = \r$ and $\alpha = 1$ for interpretation 5.
    Let $\alpha_i \geq 0$ be the amount of $\r$ purchased from cost function $i$, and $\beta_i \in \reals$ the total cost, so that the net trade from cost function $i$ is $\r^i = \alpha_i \r - \beta_i \ones$.
    Then we have $\sum_i \alpha_i = 1$, so that the net trade is $\r - \beta \ones$ where $\beta = \sum_i \beta_i$.
    By definition of $\r^i$, the cost of trades, and the $\ones$-invariance of $C_i$, we have $C_i(\q^i + \r^i) = C_i(\q^i)$.

    If $\beta = 0$ we are done.
    Otherwise, as $\r$ is Pareto optimal from part (1) above, we must have $\beta > 0$.

    From part (3) of the proof below, there exists a set of trades $\{\hat \r^i\}_i$ with $\sum_i \hat\r^i = \r$ such that $C_i(\q^i + \hat \r^i) = C_i(\q^i) = C_i(\q^i + \r^i - \beta_i\ones)$ for all $i$.
    Thus, from state $\{\hat{q}^i\}_i := \{\q^i + \r^i - \beta_i\ones\}_i$, the trades $\{\hat{r}^i\}_i := \{\hat \r^i - \r^i\}_i$ are an arbitrage, with $\sum_i \hat{r}^i = \r - \r = \0$ and net cost
    \begin{align*}
      \sum_i C_i(\hat{q}^i + \hat{r}^i) - C_i(\hat{q}^i)
      &= \sum_i C_i(\q^i + \hat \r^i - \beta_i\ones) - C_i(\q^i + \r^i - \beta_i\ones)
      \\
      &= \sum_i C_i(\q^i + \hat \r^i - \beta_i\ones) - C_i(\q^i + \hat \r^i)
      \\
      &= -\sum_i \beta_i = -\beta < 0~.
    \end{align*}

    For the converse, let $\r$ be any net trade from the continuous trading process, and $\beta \leq 0$ the optimal net cost from any arbitrage.
    By part (1) and the argument above, we have $\beta = C(\q) - C(\q+\r)$;
    taking $\hat\r = \r + (C(\q) - C(\q+\r))\ones$,
    
    we again split $\hat{\r}$ into $\{\hat{\r}^i\}_i$ and construct an arbitrage $\{\r^i\,'\}_i := \{\hat{\r}^i - \r^i\}_i$ which achieves net cost $C(\q + \r) - C(\q) = -\beta$.
    
  \item[($3\leftrightarrow 1$)]
    Let $\r$ such that $C(\q + \r) = C(\q)$.
    We must show that there exist $\{\r^i\}_i$ such that $C_i(\q^i + \r^i) = C_i(\q^i)$ and $\r = \sum_i \r^i$.
    From the definition of infimal convolution, $C(\q + \r) = \inf \{\sum_i C_i(\v^i) \mid \sum_i \v^i = \q+\r\}$.
    By Lemma~\ref{lem:infimal-conv-consistent-price}, this infimum is attained by some $\{\v^i\}_i$.
    Define $\r^i := \v^i - \q^i + (C_i(\q^i) - C_i(\v^i))\ones$.
    For the first condition, $C_i(\q^i + \r^i) = C_i(\v^i  + (C_i(\q^i) - C_i(\v^i))\ones) = C_i(\q^i)$.
    For the second,
    \begin{align*}
      \sum_i \r^i
      &= \sum_i \v^i - \sum_i \q^i + \sum_i (C_i(\q^i) - C_i(\v^i))\ones
      \\
      &= (\q + \r) - \q + \left(\sum_i C_i(\q^i) - \sum_i C_i(\v^i)\right)\ones
      \\
      &= \r + (C(\q) - C(\q + \r))\ones = \r~.
    \end{align*}
    Coherence again follows from Lemma~\ref{lem:infimal-conv-consistent-price}.
  \item[($4\leftrightarrow 1$)]
    We will show equivalence to interpretation 1.
    
    Let $\alpha_i = C_i(\q^i)$ for all $i$, so that $C(\q^i - \alpha\ones) = 0$.
    By Lemma~\ref{lem:cost-scoring-duality}, we may therefore write $\q^i = S(\p^i,\cdot) + \alpha_i \ones$.
    From Lemma~\ref{lem:infimal-conv-consistent-price}, we again have $C(\q) = \sum_i C(\q^i)$.
    From Lemma~\ref{lem:cost-scoring-duality} again, and $\ones$-invariance,
    \begin{align*}
      &\{\r^i\in\reals^n \mid C_i(\q^i+\r^i) = C_i(\q^i)\}
      \\
      &= \{\r^i\in\reals^n \mid C_i(\q^i+\r^i - \alpha_i\ones) = C_i(\q^i - \alpha_i\ones)\}
      \\
      &= \{\r^i\in\reals^n \mid C_i(S_i(\p^i,\cdot)+\r^i) = C_i(S_i(\p^i,\cdot))\}
      \\
      &= \{\r^i\in\reals^n \mid C_i(S_i(\p^i,\cdot)+\r^i) = 0\}
      \\
      &= \{\hat\q^i - S(\p^i,\cdot) \mid C_i(\hat\q^i) = 0\}
      \\
      &= \{S(\hat\p^i,\cdot) - S(\p^i,\cdot) \mid \hat\p^i\in\relint\Delta_n\}~.
    \end{align*}
    We conclude that the possible trades $\{\r^i\}_i$ in interpretation 1, such that $C_i(\q^i + \hat\r^i) = C_i(\q^i)$ for all $i$, are exactly the same as the trades $\{S_i(\hat\p^i,\cdot) - S_i(\p^i,\cdot)\}_i$ allowed in interpretation 4; we have simply reparameterized the trades by $\{\hat\p^i\}_i$.

    \item[($4\leftrightarrow 3$)] 
    We will show equivalence to interpretation 3.
    In interpretation 3, suppose $k$ market makers use scoring rules $S_{G_i}(\p,\cdot)$ for some corresponding set of generating functions $\{G_i\}_{i=1}^k$.
    Assuming a coherent market state, all market makers maintain the same initial price $\p$.
    The trader chooses a price $\hat{p}$, and receives the trade $\r^i=S_{G_i}(\hat{\p},\cdot)-S_{G_i}(\p,\cdot)$ from market maker $i$.
    Therefore, for a given initial price $\p$ and final price $\hat{\p}$, the trader receives the trade
    \[
    \r = \sum_{i=1}^k \r^i=\sum_{i=1}^k S_{G_i}(\hat{\p},\cdot)-S_{G_i}(\p,\cdot)
    \]
    Let $G=\sum_{i=1}^k G_i$, and let $S_G(\p,\cdot)$ be the corresponding scoring rule.
    It follows that $S_G(\p,\cdot)=\sum_{i=1}^k S_{G_i}(\p,\cdot)$.
    Now, considering interpretation 4, let the market maker pick generating function $G$ and thus scoring rule $S_G$, and the same initial price $\p$ as above.
    A trader chooses a new price $\hat{p}$, and receives trade $\r=S_G(\hat{\p},\cdot)-S_G(\p,\cdot)$, the same as in interpretation 3.
  \end{enumerate}
  \vspace*{-17pt}
\end{proof}

\subsection{Equivalence of the full protocols and practical considerations}
\label{sec:equiv-full-prot}

Theorem~\ref{thm:equivalence-interpretations} tells us that the process of trading in Protocols~\ref{alg:general-n-asset-cost} and~\ref{alg:general-n-asset-dual-price} are the same.
The equivalence of the rest of the protocol follows from Lemma~\ref{lem:cost-scoring-duality}, as $\liabilityf(C) = S_G(\p,\cdot)$ where $\p\in\partial C(\q)$.

The two-outcome case is similar; we need only verify the translation from $G$ and $C$ to thei 1-dimensional counterparts.
Letting $G(\p) = g(p_1)$, we have $\overline G(\x) = (x_1+x_2) g(x_1/(x_1+x_2))$ which is differentiable.
Lemma~\ref{lem:smooth-G-unique-subgradients-mod-ones} now gives $\partial G(\p) = \{(g'(p_1),0) + \alpha\ones \mid \alpha \in \reals\}$ and thus $S_G(\p,\cdot) = \dG_\p + (G(\p) + \inprod{\dG_\p}{\p})\ones = (g'(p_1),0) + (g(p_1) - p_1 g'(p_1))\ones = \liabilityf(g,p_1)$.
Computing the conjugate, we have
\begin{align*}
  G^*(\q)
  &= \sup_{\p\in\Delta_2} \inprod{\p}{\q} - G(\p)
  \\
  &= \sup_{p\in[0,]1} p q_1 + (1-p) q_2 - g(p)  
  \\
  &= \left(\sup_{p\in[0,1]} p (q_1 - q_2) - g(p)\right) + q_2
  \\
  &= g^*(q_1 - q_2) + q_2~,
\end{align*}
as desired.
Finally, to verify $\pricef(\cdot)$, note that $c' = (g')^{-1}$ whenever both derivatives are defined.
By assumption on $\Ginit$, any argument to $\pricef$ is both differentiable and strictly convex, and thus $c$ is differentiable.
We have now established Proposition~\ref{prop:equivalence-two-asset-general}.

\section{Proofs and additional work from Section~\ref{sec:fees}} \label{sec:fees-appendix}

\subsection{Worked example for Uniswap fees} \label{sec:app-uniswap-fee}

Consider a market with two LPs $G=G_1+G_2$ where $G_1(\p)=-2\sqrt{p_1p_2}$ and $G_2(\p)=-2\sqrt{p_2p_3}$. 
Per Protocol~\ref{alg:general-n-asset-dual-price}, a trade $\p \to\ \hat{\p}$ is given by $\r = S_G(\hat{\p},\cdot) - S_G(\p,\cdot)$, where
\begin{align*}
  S_G(\p,\cdot) = \nabla \overline G(\p) = \nabla \overline G_1(\p) + \nabla \overline G_2(\p) = \left(\sqrt{p_1/p_2},\sqrt{p_2/p_1}+\sqrt{p_2/p_3},\sqrt{p_3/p_2}\right)~.
\end{align*}
The trade is then split among the two LPs, as
\begin{align*}
  \r^1 &= \nabla G_1(\hat{\p}) - \nabla G_1(\p) = \left(\sqrt{\hat{p}_1/\hat{p}_2}-\sqrt{p_1/p_2},\sqrt{\hat{p}_2/\hat{p}_1}-\sqrt{p_2/p_1},0\right)~,
  \\
  \r^2 &= \nabla G_2(\hat{\p}) - \nabla G_2(\p) = \left(0,\sqrt{\hat{p}_2/\hat{p}_3}-\sqrt{p_2/p_3},\sqrt{\hat{p}_3/\hat{p}_2}-\sqrt{p_3/p_2}\right)~.
\end{align*}
Start the market at the uniform price $\p = (1/3,1/3,1/3)$, and consider a trader wishing to purchase security 1 in exchange for security 3, as above.
Intuitively, as there is liquidity between securities 1 and 2 (provided by LP 1) and between securities 2 and 3 (provided by LP 2), there should be ``combined'' liquidity between 1 and 3.
And indeed that is the case: if the trader selects $\hat{\p} = \left( \frac{3/2 + \sqrt{2}}{3 + \sqrt{2}}, \frac{1}{3 + \sqrt{2}}, \frac{1}{2(3 + \sqrt{2})} \right)$, an expression chosen for arithmetic convenience, we have a resulting trade $\r = \nabla \overline G(\hat{\p}) - \nabla \overline G(\p) = (\sqrt{3/2}-1,0,\sqrt{1/2}-1)$.
The split $\r = \r^1 + \r^2$ between the LPs is also roughly as one would expect, each $\r^i$ being between the corresponding pair of securities:
\begin{align*}
  \r^1 &= \nabla G_1(\hat{\p}) - \nabla G_1(\p) = \left(0,\sqrt{2}-1,\sqrt{1/2}-1\right)~,
  \\
  \r^2 &= \nabla G_2(\hat{\p}) - \nabla G_2(\p) = \left(\sqrt{3/2}-1,1-\sqrt{2},0\right)~.
\end{align*}

Using Uniswap's fee, i.e., $\fee(\r,\q)=\beta (-\r)_+$, $\fee_i(\r,\q)=\beta (-\r^i)_+$, presents an issue.
The fee charged to the trader, $\beta(-\r)_+ = \beta(0,0,1-\sqrt{1/2})$, ignores the fact that LP 2 provided liquidity that facilitated the trade!
Indeed, looking at the fees paid to LPs, we see this same fee $\beta(-\r^1)_+ = \beta(0,0,1-\sqrt{1/2})$ paid to LP 1, plus an additional fee of $\beta(-\r^2)_+ = \beta(0,\sqrt{2}-1,0)$ to LP 2.

Fortunately, this issue is not present in 2-asset protocols like Protocol~\ref{alg:general-two-asset-g}, but clearly it can emerge beyond 2 assets/outcomes.
Moreover, it seems to emerge precisely when there is ``synergy'' among the LPs, enabling trades that fruitfully combine their liquidity.
While one could easily fix this issue by directly charging the trader for the sum of the fees to the LPs, doing so may be problematic.
For example, this proposal would break the abstraction barrier, in the sense that the fees would depend intimately on the LP profile, not just their combined liquidity.

\subsection{Proof of Lemma~\ref{lem:feeuni}}

\setcounter{lemma}{0}
\begin{lemma} 
    Axioms~\ref{axi:bb}, \ref{axi:ts}, and \ref{axi:ld} allow us to write 
    \begin{equation*}
        \overline{\fee}_\mathrm{T}=\fee_\mathrm{LP}.
    \end{equation*}
    \begin{proof}
    Assume without loss of generality, by Theorem~\ref{thm:equivalence-interpretations}, that the market maker has only one LP providing all the liquidity.
    Hence $\r=r^1$ and $\q=q^1$.
    Then, Axioms~\ref{axi:bb},~\ref{axi:ld} give us that 
    \begin{align*}
        \fee_{1}(\{\vec{r}^1\},\{\vec{q}^1\})&=\fee_{\mathrm{T}}(\{\vec{r}^1\},\{\vec{q}^1\}) \\
        &=\fee_{\mathrm{LP}}(\r^1,\q^1)
    \end{align*}
    Moreover, Axiom \ref{axi:ts} says
    \begin{align*}
        \fee_{\mathrm{T}}(\{\r^1\},\{\q^1\})&=\overline\fee_{\textsc{T}}(\r^1,\q^1)
    \end{align*}
    Hence $\overline\fee_{\textsc{T}}() = \fee_{\mathrm{LP}}()$.
    \end{proof}
\end{lemma}

\subsection{Proof of Theorem~\ref{thm:fee-impossible}}

\setcounter{theorem}{1}
\begin{theorem} 
Axioms~\ref{axi:bb}, \ref{axi:ts}, \ref{axi:ld}, and \ref{axi:nn} are incompatible.
\end{theorem}
\begin{proof}
\newcommand{\ppost}{\p_{\mathrm{post}}}
\newcommand{\qpost}{\q_{\mathrm{post}}}
\newcommand{\bppost}{\bar{\p}_{\mathrm{post}}}
\newcommand{\bqpost}{\bar{\q}_{\mathrm{post}}}
We will consider the scoring rule interpretation of liquidity provisioning and parallel market making, according to Protocol 3. Consider $n=3$ assets, and 5 LPs, which provide liquidity according to the following 5 generating functions:
\begin{align*}
    & G_1 = -2\sqrt{2p_2p_3}, \\
    & G_2 = -2\sqrt{2p_1p_3}, \\
    & G_3 = -2\sqrt{2p_1p_2}, \\
    & G_4 = 49\cdot\frac{p_3^2-p_3}{21p_3+4}, \\
    & G_5 = 9(p_3^2-p_3).
\end{align*}
Note that $G_1,G_2,$ and $G_3$ are 1-homogeneous generating functions, each symmetric with respect to two assets and flat with respect to the third, and equal to each other up to permutation of the assets.
$G_4$ and $G_5$ are designed to provide specific liquidity vectors in certain scenarios while also not facilitating any part of a particular trade.
Now consider the following market.
Set an initial price vector of $\p=(1/7,4/7,2/7)$.
Let LPs LP$_1$, LP$_2$, and LP$_4$ enter the market with generating functions $G_1$, $G_2$, and $G_4$, respectively. The individual liabilities $\q^i$ for each LP $i$ according to $S_{G_i}(\p,\cdot)$ are
\begin{align*}
& \q^1 = (0, -1, -2),   \\
& \q^2 = (-2, 0, -1),\\
& \q^4 = (-1,-1,-1).
\end{align*}
And the total liability is \[\q=\q^1+\q^2+\q^4=(-3,-2,-4)\]
according to $S_{G_1+G_2+G_4}(\p,\cdot)$.
Now, we trade to the price $\ppost=(4/7,1/7,2/7)$.
At price $\ppost$, the liabilities according to the protocol are
\begin{align*}
    & \qpost^1 = (0, -2, -1), \\
    & \qpost^2 = (-1, 0, -2),\\
    & \qpost^4 = (-1, -1, -1), \\
    & \qpost = (-1, -2, -3).
\end{align*}
Therefore, the net trades facilitated by each LP are
\begin{align*}
    & \r^1 = (0, -1, 1), \\
    & \r^2 = (1, 0, -1),\\
    & \r^4 = (0, 0, 0).
\end{align*}
And the total net trade is \[\r=\r^1+\r^2+\r^4=(1,-1,0).\]
Then, the fees assessed are
\begin{equation}
\begin{aligned} \label{eq:1}
    &\fee(\r,\q) \\
    =&\fee((1,-1,0),(-3,-2,-4))\\
    =& \fee_1(\r^1,\q^1)+\fee_2(\r^2,\q^2)+\fee_4(\vec{0},\q^4) \\
    =& \fee_1((0,-1,1),(0,-1,-2))+\fee_2((1,0,-1),(-2,0,-1))
\end{aligned}
\end{equation}
Note the application of Axiom \ref{axi:nn}, which states that the fee for the trade of $\r^4=\mathbf{0}$ must be 0.

Now, consider the following alternative scenario.
We set an initial price of $\p=(2/9, 4/9, 1/3)$.
LPs LP$_3$ and LP$_5$ enter the market with generating functions $G_3$ and $G_5$, respectively.
The liabilities deposited are
\begin{align*}
& \q^3 = (-2, -1, 0), \\
& \q^5 = (-1, -1, -4), \\
& \q = (-3, -2, -4).
\end{align*}
Now, we trade to the price $\ppost=(4/9,2/9,1/3)$.
At price $\ppost$, the liabilities according to the protocol are
\begin{align*}
    & \qpost^3 = (-1, -2, 0), \\
    & \qpost^5 = (-1, -1, -4), \\
    & \qpost = (-2, -3, -4).
\end{align*}
And the net trades are therefore
\begin{align*}
    & \r^3 = (1, -1, 0), \\
    & \r^5 = (0, 0, 0), \\
    & \r = (1, -1, 0).
\end{align*}
The fee assessed is
\begin{equation}
\begin{aligned} \label{eq:2}
    &\fee(\r,\q) \\
    =& \fee((1,-1,0),(-3,-2,-4))\\
    =& \fee_3(\r^3,\q^3)+\fee_5(\r^5,\q^5) \\
    =& \fee_3((1,-1,0),(-1,-2,0))+\fee_5(\vec{0},(-1,-1,-4)) \\
    =& \fee_3((1,-1,0),(-1,-2,0))
\end{aligned}
\end{equation}
Note that the fee for $\r^5=\mathbf{0}$ is 0, per Axiom \ref{axi:nn}.
Combining eqs.~\eqref{eq:1} and \eqref{eq:2}, we have that
\begin{equation}
\begin{aligned} \label{eq:3}
\fee_1(\r^1,\q^1)+\fee_2(\r^2,\q^2)=\fee_3(\r^3,\q^3)
\end{aligned}
\end{equation}

Now, we consider a scenario that is identical up to permutation of assets, with assets 2 and 3 being swapped.
We introduce two more generating functions $G_6$ and $G_7$, which are identical to $G_4$ and $G_5$ but parameterized by $p_2$ rather than $p_3$.
\begin{align*}
    & G_4 = 49\cdot\frac{p_2^2-p_2}{21p_2+4},\\
    & G_5 = 9(p_2^2-p_2).
\end{align*}
Let LPs LP$_1$, LP$_3$, and LP$_6$ enter the market with generating functions $G_1$, $G_3$, and $G_6$ respectively.
We set an initial price vector of $\bar{\p}=(1/7, 2/7, 4/7)$, and trade to a price vector of $\bppost=(4/7,2/7,1/7)$.
The initial liabilities $\bar{\q}^i$, trades $\bar{\r}^i$, and final liabilities $\bqpost^i$ are as follows:
\begin{align*}
    & \bar{\q}^1 = (0, -2, -1), \\
    & \bar{\q}^3 = (-2, -1, 0), \\
    & \bar{\q}^6 = (-1, -1, -1), \\
    & \bar{\q} = (-3, -4, -2), \\
    & \\
    & \bar{\r}^1 = (0, 1, -1),  \\
    & \bar{\r}^3 = (1, -1, 0), \\
    & \bar{\r}^6 = (0, 0, 0), \\
    & \bar{\r} = (1, 0, -1), \\
    & \\
    & \bqpost^1 = (0, -1, -2), \\
    & \bqpost^3 = (-1, -2, 0), \\
    & \bqpost^6 = (-1, -1, -1), \\
    & \bqpost = (-2, -4, -3).
\end{align*}
The fees assessed are
\begin{equation}
\begin{aligned} \label{eq:4}
    &\fee(\bar{\r},\bar{\q}) \\
    =& \fee((1,0,-1),(-3,-4,-2))\\
    =& \fee_1(\bar{\r}^1,\bar{\q}^1)+\fee_3(\bar{\r}^3,\bar{\q}^3)+\fee_6(\vec{0},\bar{\q}^6) \\
    =& \fee_1((0,1,-1),(0,-2,-1))+\fee_3((1,-1,0),(-2,-1,0)) \\
\end{aligned}
\end{equation}

Now, instead consider a market with LPs LP$_2$ and LP$_7$, with generating functions $G_2$ and $G_7$ respectively, an initial price of $\bar{\p}=(2/9, 1/3, 4/9)$, and a final price of $\bppost=(4/9, 2/9, 1/3)$.
The corresponding liabilities and trades are
\begin{align*}
    & \bar{\q}^2 = (-2, 0, -1), \\
    & \bar{\q}^7 = (-1, -4, -1), \\
    & \bar{\q} = (-3, -4, -2), \\
    & \\
    & \bar{\r}^2 = (1, 0, -1), \\
    & \bar{\r}^7 = (0, 0, 0), \\
    & \bar{\r} = (1, 0, -1), \\
    & \\
    & \bqpost^2 = (-1, 0, -2), \\
    & \bqpost^7 = (-1, -4, -1), \\
    & \bqpost = (-2, -4, -3).
\end{align*}
The fees assessed are
\begin{equation}
\begin{aligned} \label{eq:5}
    &\fee(\bar{\r},\bar{\q}) \\
    =& \fee((1,0,-1),(-3,-4,-2))\\
    =& \fee_2(\bar{\r}^2,\bar{\q}^2)+\fee_7(\bar{\r}^7,\bar{\q}^7) \\
    =& \fee_2((1,0,-1),(-1,0,-2))+\fee_7(\vec{0},(-1,-4,-1)) \\
    =& \fee_2((1,0,-1),(-1,0,-2))
\end{aligned}
\end{equation}
Combining eqs.~\eqref{eq:4} and \eqref{eq:5}, we have that
\begin{equation}
\begin{aligned} \label{eq:6}
\fee_1(\bar{\r}^1,\bar{\q}^1)+\fee_3(\bar{\r}^3,\bar{\q}^3) = \fee_2(\bar{\r}^2,\bar{\q}^2) \\
\end{aligned}
\end{equation}
From eqs.~\eqref{eq:3} and \eqref{eq:6}, we have that
\begin{align*}
\fee_1(\r^1,\q^1)+\fee_2(\r^2,\q^2) &= \fee_3(\r^3,\q^3), \\
\fee_1(\bar{\r}^1,\bar{\q}^1)+\fee_3(\bar{\r}^3,\bar{\q}^3) &= \fee_2(\bar{\r}^2,\bar{\q}^2).
\end{align*}
Observe that $\bar{\r}^2 = \r^2$, $\bar{\r}^3 = \r^3$, $\bar{\q}^2 = \q^2$, and $\bar{\q}^3 = \q^3$, giving us
\begin{align*}
\fee_1(\r^1,\q^1)+\fee_2(\r^2,\q^2) &= \fee_3(\r^3,\q^3), \\
\fee_1(\bar{\r}^1,\bar{\q}^1)+\fee_3(\r^3,\q^3) &= \fee_2(\r^2,\q^2).
\end{align*}
Because $\r^1\neq\0$, by Axiom~\ref{axi:nn}, $\fee_1(\r^1,\q^1)>0$, and therefore $\fee_2(\r^2,\q^2) < \fee_3(\r^3,\q^3)$.
Similarly, because $\bar{\r}^1 \neq\0$, $\fee_1(\bar{\r}^1,\bar{\q}^1)>0$, and therefore $\fee_3(\r^3,\q^3) < \fee_2(\r^2,\q^2)$.
But then $\fee_2(\r^2,\q^2) < \fee_3(\r^3,\q^3)$ and $\fee_3(\r^3,\q^3) < \fee_2(\r^2,\q^2)$, a contradiction.
\end{proof}

\section{Proofs and additional work from Section~\ref{sec:DeFi}}\label{sec:omitted_proofs-defi}

\subsection{General protocol for the exchange of two securities} \label{sec:app-general-two-asset}

Before we analyze Uniswap, we introduce a simpler version of Protocol~\ref{alg:general-n-asset-cost} for the two outcome case. This is our Protocol~\ref{alg:general-two-asset-g}.
This helps us reason about liquidity functions more conveniently as they apply to decentralized finance, as only two securities are exchanged.
Recall that we use lowercase $g$ and $c$ for the two outcome case instead of $G$ and $C$; see \S~\ref{sec:liquidity-scalar}.
We detail the specifications in the next paragraph.

Since $C(S_g(p,\cdot))=0$ and $\nabla C(S_g(p,\cdot))=(p,1-p)$, we can replace $\liabilityf(C)$ on Line 3 of Protocol~\ref{alg:general-n-asset-cost} with $\liabilityf(g,p)=S_g(p,\cdot)$ where $p$ is the current price of asset 1. 
Similarly, we can more easil compute the trade check on Line 14, and optimal split on Line 17 in Protocol~\ref{alg:general-n-asset-cost} using $S_g(p,\cdot)$ without needing to compute infimal convolutions explicitly.

The formal proof of equivalence of Protocol~\ref{alg:general-n-asset-cost} and Protocol~\ref{alg:general-two-asset-g} follows from the equivalence between the corresponding $n$-asset protocols (Theorem~\ref{thm:equivalence-interpretations}).
Let $\G$ be the set of functions $g:[0,1]\to\reals_{\leq 0}$ which are convex, continuously differentiable, and bounded.
Let $\G^* \subseteq \G$ be those which additionally have $|g'(p)| \to \infty$ as $p\to 0$ or $p\to 1$.

\begin{proposition}\label{prop:equivalence-two-asset-general}
 Let $\Ginit \subseteq \G^*$ be a set of functions which are strictly convex.
  Then for $n=2$ Protocol~\ref{alg:general-n-asset-cost} is equivalent to Protocol~\ref{alg:general-two-asset-g} for the choices $c_i = g_i^*$ where $C_i(\q) = c_i(q_1-q_2) + q_2$.
\end{proposition}

\begin{algorithm}[t]
  \caption{General two-asset protocol via liquidity functions}
  \label{alg:general-two-asset-g}
  \begin{algorithmic}[1]
    \State \textbf{global constant} $\Ginit \subseteq \G^*$, $\GLP \subseteq \G$, $\fee(),\fee_i() \in\reals_{\geq 0}$
    \State \textbf{global variables} $k\in\N$, $\{\q^i\in\reals^2\}_{i=0}^k$, $\{g_i\in\GLP\}_{i=0}^k$
    \State $\pricef(g,\q) := (g')^{-1}(q_1-q_2)$
    \State $\liabilityf(g,p) := S_g(p,\cdot)$ \hfill where $ S_g(p,\cdot) = (g'(p),0) + (g(p)-p\cdot g'(p))\ones$.
    \label{item:two-asset-g-liabilityf}
    \medskip
    
    \Function{Initialize}{$\q \in \reals^2,g\in\Ginit$} \Comment{Equivalently, the market creator can specify $\ell = g''$}
    \State $(k,\q^0,g_0) \gets (0,\q,g)$
    \State \textbf{check} $\q^0 = \liabilityf(g_0,\pricef(g_0,\q))$
    \EndFunction
    \medskip

    \Function{RegisterLP}{$i=k+1$}
    \State $(k,\q^i,g_i) \gets (k+1,0,0)$
    \EndFunction
    \medskip

    \Function{ModifyLiquidity}{$i \in \N, g \in \GLP$} \Comment{Equivalently, the LP can specify $\ell = g''$}
    
    \State \textbf{request} $\r^i = \q^i - \liabilityf(g,\pricef(g_0,\q^0))$ from LP $i$
    
    \label{item:request_q_general_g}
    \State $(\q^i,g_i) \gets (\q^i - \r^i,g)$
    \EndFunction
    \medskip
    
    \Function{ExecuteTrade}{$\r \in \reals^2$}

    \State $\hat{p} \gets \pricef(\sum_{j=0}^kg_j,\sum_{i=0}^k \q^i+\r)$ \Comment{The price after this trade}
    \State \textbf{check} $\sum_{i=0}^k \q^i+\r = \liabilityf(\sum_{j=0}^kg_j,\hat{p})$
    \label{item:general_aggregate_check}
    \State \textbf{trader pays} $\fee(\r,\q,g)$
    
    cash in fee
    \State\textbf{ Give} $r$ to trader
    \For{each LP $i$}
    \State $\r^i \gets  \liabilityf(g_i,\hat{p}) - \q^i$.
    \label{item:check_trade_general_g}
    \State LP $i$ gets $\fee_i(\r,\q,\{g_i\}_{i=1}^k)$ fees
    \label{item:two-asset-fee-split}
    \State $\q^i \gets \q^i + \r^i$
    \EndFor
    \EndFunction
  \end{algorithmic}
\end{algorithm}

\subsection{Proofs related to Uniswap V2 in Section \ref{sec:DeFi}} \label{sec:app-uniswapv2}

In this section, we show that Uniswap V2, outlined in Protocol~\ref{alg:V2-two-asset}, is a special case of Protocol~\ref{alg:general-two-asset-g}.

\begin{algorithm}[ht]
  \caption{Uniswap V2}
  \label{alg:V2-two-asset}
  \begin{algorithmic}[1]
    \State \textbf{global constants} $\beta$.
    \State \textbf{global variables} $\res\in\reals^2$, $k\in\N$, $\{\alpha^i\in\reals_{\geq 0}\}_{i=0}^k$.
    \State $\pricef(\res) := \frac{\xb}{\xa+\xb}$
    \medskip
    
    \Function{Initialize}{$\res^0 \in \reals^2,\beta$}
    \State $(k,\beta) \gets (0,\beta)$
    \State $\alpha^0 \gets \sqrt{\xa^0\cdot \xb^0}$
    \EndFunction
    \medskip
    
    \Function{RegisterLP}{$i=k+1$}
    \State $(k,q^i,\alpha^i) \gets (k+1,0,0)$ 
    \EndFunction
    \medskip
    
    \Function{ModifyLiquidity}{$i \in \N, \alpha' \geq 0$}
    \State $p=\textproc{price}(\res)$
    
    \State \textbf{request }$\res' = \left((\alpha'-\alpha^i)\sqrt{\frac{1-p}{p}},(\alpha'-\alpha^i)\sqrt{\frac{p}{1-p}}\right)$ from LP $i$.
    \label{item:request-q-v2}
    \State $(\res,\alpha^i) \gets (\res+\res',\alpha')$ 
    \EndFunction
    \medskip
    
    \Function{ExecuteTrade}{$\r \in \reals^2$}
    \label{item:V2-trade-check}
    \State \textbf{check} $\varphi(\res)=\varphi(\res-\r)$, where $\varphi(\res)=\xa\xb$.
    \State \textbf{pay} $\frac{\beta\alpha^i}{\alpha}(-\r)_+$ to LP $i$, where $\alpha=\sum_i \alpha^i$.
    \label{item:V2-fee-split}
    \State $\res\gets\res-\r $. 
    \EndFunction
  \end{algorithmic}
\end{algorithm}

\begin{proposition}\label{prop:is_uniswap} 
  For $\Ginit = \{\alpha g_0 \mid \alpha > 0\}$, $\GLP = \{\alpha g_0 \mid \alpha \geq 0\}$ where $g_0(p)=
  -2\sqrt{p(1-p)}$ in Protocol \ref{alg:general-two-asset-g}, $\liabilityf(g,p) = -\alpha \left(\sqrt{\frac{1-p}{p}}, \sqrt{\frac{p}{1-p}}\right)^\top$ for $g=\alpha g_0$. The vector $\x$ of reserves in Protocol \ref{alg:V2-two-asset} satisfies $\xa\cdot \xb = \alpha^2$ for some $\alpha > 0$ iff $\q=\liabilityf(g,\pricef(g,p))$, where $\res = - \q$.
\end{proposition}
\begin{proof}
  Observe that any change in liability vector in the \textsf{ModifyLiquidity} or \textsf{Initialize} phases results in a liability vector that takes the form $\liabilityf(g,p)$ for some $g\in \GLP$, $p\in[0,1]$. 
  Let this $g(p)=\alpha g_0(p)=-2\alpha\sqrt{p(1-p)}$ and thereby $g'(p) = -\alpha\frac{1-2p}{\sqrt{p(1-p)}}$. So,
  \begin{align*}
    \liabilityf(g,p)&= (g(p)-p\cdot g'(p)\ones + \begin{pmatrix}
      g'(p) \\ 0
    \end{pmatrix}\\
                    &= \alpha\left(\left(-2\sqrt{p(1-p)}+p\frac{1-2p}{\sqrt{p(1-p)}}\right)\ones-\begin{pmatrix}\frac{1-2p}{\sqrt{p(1-p)}}\\0\end{pmatrix}\right)= \alpha\begin{pmatrix}-\sqrt{\frac{1-p}{p}}\\-\sqrt{\frac{p}{1-p}}\end{pmatrix}~.
  \end{align*}
  Hence if $\q=\liabilityf(g,\pricef(g,p))$, 
  \begin{align*}
      \begin{pmatrix}
          q_1 \\ q_2
      \end{pmatrix} ~= \liabilityf(g,\pricef(g,p)) &= -\alpha \begin{pmatrix}\sqrt{\frac{1-\pricef(g,p)}{\pricef(g,p)}}\\\sqrt{\frac{\pricef(g,p)}{1-\pricef(g,p)}}\end{pmatrix}
  \end{align*}
  which implies $q_1\cdot q_2 = \alpha^2 = x_1\cdot x_2$.

  For the if direction, price of market at $q$ is given by $(g')^{-1}(q_1-q_2)$, call this $p$.
    \begin{align*}
         q_1-q_2 &= g'(p) =
        -\alpha\frac{1-2p}{\sqrt{p(1-p)}}
    \end{align*}
  Solving this, we get that $p=\frac{q_2}{q_1+q_2}$. 
  \begin{align*}
      \liabilityf(g,\pricef(g,p)) &= \liabilityf(g,\frac{q_2}{q_1+q_2})\\
      &= \alpha\begin{pmatrix}-\sqrt{\frac{q_1}{q_2}}\\-\sqrt{\frac{q_2}{q_1}}\end{pmatrix} =  \begin{pmatrix}q_1\\q_2\end{pmatrix}.
  \end{align*}
  The last line uses the fact that $q_1\cdot q_2 = \alpha^2$.
\end{proof}
\begin{proposition}\label{prop:protocol_equal} 
  Protocol \ref{alg:V2-two-asset} is equivalent to Protocol \ref{alg:general-two-asset-g} for $\Ginit = \{\alpha g_0 \mid \alpha > 0\}$, $\GLP = \{\alpha g_0 \mid \alpha \geq 0\}$ where $g_0(p)=-2\sqrt{p(1-p)}$, $\fee(\r,\q,g)=\beta\r_+$ and $\fee_i(\r,\q,\{g_i\}_i)=\frac{\beta\alpha_i}{\alpha}\r_+$ for $\beta >0$.
\end{proposition}
\begin{proof}
  Showing that Protocol \ref{alg:general-two-asset-g} gives Protocol \ref{alg:V2-two-asset}, for the choices of $g$ specified, involves showing the equivalence of three specific components.
  That is, we wish to show that the initialization check in Line 7 of Protocol \ref{alg:general-two-asset-g} is satisfied and that the request vectors in \textsf{ModifyLiquidity} and \textsf{ExecuteTrade} routines match.

  First, we note that modifying $\x$ preserves the invariant $\xa\cdot \xb=\alpha^2$ for some $\alpha$ in Protocol \ref{alg:V2-two-asset}. As shown in the proof of Proposition \ref{prop:is_uniswap}, $ (g')^{-1}(q_1-q_2) = \pricef(g_0,\q) = \frac{q_2}{q_1+q_2}$ where $\q=-\res$, showing that the Uniswap price function is a special case.

  For the initialization phase, $\xa^0\cdot \xb^0=(\alpha^0)^2$ and Proposition $\ref{prop:is_uniswap}$ give us that $\liabilityf(g_0,\pricef(g_0,q^0))=\q^0$ hence satisfying the check.
  By induction, this statement holds for all states $\q$ i.e. $\q=\liabilityf(g,\pricef(g,\q))$.

  A liquidity change of $\alpha^i$ to $\alpha'$ at price $p$ reflects a change of $g_i$ from $-2\cdot\alpha^i\sqrt{p(1-p)}$ to $\hat{g}=-2\cdot\alpha'\sqrt{p(1-p)}$ in Uniswap V2.
  The quantity of assets requested by Protocol \ref{alg:general-two-asset-g} is given by
  \begin{align*}
    \r^i=-(\liabilityf(\hat g,p) - \q^i)&=-(\liabilityf(\hat g,p) - \liabilityf(g_i,p)) \\
                                   &= (\alpha'-\alpha^i)\begin{pmatrix}\sqrt{\frac{1-p}{p}}\\\sqrt{\frac{p}{1-p}}\end{pmatrix} =x'
  \end{align*}
  Now, we show that the check in Line \ref{item:check_trade_general_g} of Protocol \ref{alg:general-two-asset-g} for the given $g$ gives us the condition $\xa\cdot \xb = (\xa-r_1)\cdot(\xb-r_2)$ that appears in Uniswap V2.

  Let $\q=\sum_i\q^i$ and $\r=\sum_i\r^i$.
  The check in Protocol \ref{alg:general-two-asset-g} can be rewritten as 
  \begin{align*}
    \q+\r&= \liabilityf\left(\sum_{j=0}^k g_j,\hat{p}\right)\\
         &= \liabilityf\left(\sum_{j=0}^k g_j,\pricef\left(\sum_{j=0}^k g_j,\q+\r\right)\right)
  \end{align*}
  From this, by applying Proposition \ref{prop:is_uniswap}, $(q_1+r_1)(q_2+r_2)=(-\xa+r_1)(-\xb+r_2)=\alpha^2=\xa\cdot \xb$.
  Again the note the difference in sign conventions for liability and reserves.

  The last fact we want to show to complete the proof is that $\r^i=\frac{\alpha^i}{\alpha} \r$ for Uniswap V2. This can be obtained by observing that $\r = \liabilityf(g,\hat{p})-\q = \liabilityf(g,\hat{p})-\liabilityf(g,p)$ from Proposition \ref{prop:is_uniswap} and $g=\frac{\alpha^i}{\alpha} g_i$ being true for Uniswap V2.

\end{proof}
\subsection{Liquidity vectors for the general bucketing mechanism} \label{sec:bucket_liquidity_vec}

In this section, we derive the expressions for the liquidity vectors for general bucketing mechanisms.
We are given, for a specified $g$, an $\ell^{(j)}$ function of the form below,
\begin{align*}
  \ell^{(j)}(p)&=g''(p)\ones_{[a_j,b_j]}(p) =\begin{cases}
    0 & p<a_j \\
    g''(p) & p\in[a_j,b_j] \\
    0 & p>b_j
  \end{cases}.
\end{align*}
Then, $(\hat{g}^{(j)})'(p)$ is given by
\begin{align*}
  (\hat{g}^{(j)})'(p)&=\int_0^p \ell^{(j)}(s)\,ds
  =\begin{cases}
    0 & p < a_j \\
    g'(p)-g'(a_j) & p\in[a_j,b_j] \\
    g'(b_j)-g'(a_j) & p>b_j
  \end{cases}.
\end{align*}
We integrate from $0$ to $p$ to see that
\begin{align*}
  \hat{g}^{(j)}(p)=\int_0^p  (\hat{g}^{(j)})'(s)ds
  &=\begin{cases}
    0 & p < a_j \\
    g(p)-g(a_j)-g'(a_j)(p-a_j) & p\in[a_j,b_j] \\
    (g'(b_j)-g'(a_j))(p-b_j)+g(b_j)-g(a_j)-g'(a_j)(b_j-a_j) & p>b_j
  \end{cases}.
\end{align*}
Then,
\begin{align*}
  g^{(j)}(p)&=\hat{g}^{(j)}(p)-p\hat{g}^{(j)}(1) \\
  &=\begin{cases}
    p(g(a_j)-g(b_j)-g'(a_j)(a_j-1)+g'(b_j)(b_j-1)) & p < a_j \\
    g(p)+g(a_j)(p-1)-pg(b_j)-a_jg'(a_j)(p-1)+pg'(b_j)(b_j-1) & p\in[a_j,b_j] \\
    (p-1)(g(a_j)-g(b_j)-a_jg'(a_j)+b_jg'(b_j)) & p>b_j
  \end{cases}.
\end{align*}
From this we can calculate the liability vector as follows
\begin{align*}
  \liabilityf(g^{(j)},p)&=(g'(p),0)+(g(p)-pg'(p))\ones\\
  &=\begin{cases}
    \begin{pmatrix} 
      g(a_j)-g(b_j)-g'(a_j)(a_j-1)+g'(b_j)(b_j-1)
      \\
      0
    \end{pmatrix} & p < a_j \\
    \begin{pmatrix}
      g(p)-g(b_j)-g'(p)(p-1)+g'(b_j)(b_j-1) \\
      g(p)-g(a_j)-pg'(p)+a_jg'(a_j)
    \end{pmatrix} & p\in[a_j,b_j] \\
    \begin{pmatrix}
      0 \\ g(b_j)-g(a_j)+a_jg'(a_j)-b_jg'(b_j)
    \end{pmatrix} & p>b_j
  \end{cases}\\
&=\begin{cases}
      \begin{pmatrix} 
      S(a_j,1)-S(b_j,1)
      \\
      0
    \end{pmatrix} & p < a_j \\
    \begin{pmatrix}
      S(p,1)-S(b_j,1) \\
      S(p,0)-S(a_j,0)
    \end{pmatrix} & p\in[a_j,b_j] \\
    \begin{pmatrix}
      0 \\ S(b_j,0)-S(a_j,0)
    \end{pmatrix} & p>b_j
  \end{cases}\\
  &=\begin{pmatrix}
      S_g(\max\{a_j,p\},1)-S_g(\max\{b_j,p\},1)
      \\
      S_g(\min\{b_j,p\},0)-S_g(\min\{a_j,p\},0)
  \end{pmatrix}
\end{align*}
where $S(p,y)=g(p)+g'(p)(y-p)$.
\subsection{A bucketing scheme for logarithmic market scoring rule (LMSR)}\label{sec:bucket_lmsr}
Here, we apply the techniques of the previous section to consider an interesting new protocol.
What if we used $g$ from LMSR but with a bucketing scheme similar to Uniswap V3? 
That is, LPs can deposit according to $g(p)=p\log p+(1-p)\log(1-p)$ on discrete buckets analogous to what we see in Uniswap V3.
We see that $g'(p)=\log p-\log(1-p)=\log\left(\frac{p}{1-p}\right)$, so
\begin{equation*}
  g^{(j)}(p)=\begin{cases}
    p\log\frac{a_j}{b_j} & p < a_j \\
    p\log \frac{p}{b_j}+(1-p)\log\frac{(1-p)}{(1-a_j)} & p\in[a_j,b_j] \\
    (1-p)\log\frac{(1-b_j)}{(1-a_j)} & p>b_j
  \end{cases}.
\end{equation*}
\begin{align*}
  \liabilityf(g^{(j)},p)&=(g'(p),0)+(g(p)-pg'(p))\ones\\
  &=\begin{cases}
    \begin{pmatrix} 
      \log\frac{a_j}{b_j}
      \\
      0
    \end{pmatrix} & p < a_j \\
    \begin{pmatrix}
      \log\frac{p}{b_j} \\
      \log\frac{(1-p)}{(1-a_j)}
    \end{pmatrix} & p\in[a_j,b_j] \\
    \begin{pmatrix}
      0 \\ \log\frac{(1-a_j)}{(1-b_j)}
    \end{pmatrix} & p>b_j
  \end{cases}.
\end{align*}
\subsection{Discussion on Uniswap V3} \label{sec:uniswapV3}

\begin{algorithm}
  \caption{Uniswap V3}
  \label{alg:V3-two-asset}
  \begin{algorithmic}[1]
    \State \textbf{global constants} $\beta$, $m$, $\{B^j = [a_j,b_j]\}_{j=0}^m$. 
    \State \textbf{global variables} $\res\in\reals^2$, $k\in\N$, $\{\alpha^{ij}\in\reals_{\geq 0}\}_{i\in \{0,\ldots,k\},j\in \{0,\cdots,m\}}$
    \Function{price}{$\res \in \reals^2$}
    \State \textbf{return } $\frac{\xb+\alpha_j\sqrt{\frac{a_j}{1-a_j}}}{\xa+\alpha_j\sqrt{\frac{1-b_j}{b_j}}+\xb+\alpha_j\sqrt{\frac{a_j}{1-a_j}}}$ 
    \EndFunction
    \Function{Initialize}{$\res \in \reals^2,\alpha>0,\beta$}
    \State $(k,\beta) \gets (0,\beta)$
    \State $\textproc{ModifyLiquidity}(0, \res, \alpha,[0,1])$
    \Comment{We technically have to split this into function call for each price interval.}
    \EndFunction
    \medskip
    
    \Function{RegisterLP}{\mbox{}}
    \State $k \gets k+1$
    \State $\alpha^{kj} \gets 0,\forall j\in\{0,\ldots,m\}$
    \State \textbf{return} $k$ \Comment{ID of the new LP }
    \EndFunction
    \medskip
    
    \Function{ModifyLiquidity}{$i \in \N, \alpha' \geq 0,j\in \{0,\ldots,m\}$}
    \State $p=\textproc{price}(\res)$
    
    \State \textbf{request }$\res' = \begin{cases} \label{item:V3-liquidity-request}
      \left((\alpha'-\alpha^{ij})\left(\sqrt{\frac{1-a_j}{a_j}}-\sqrt{\frac{1-b_j}{b_j}}\right),0\right) & \text{ if } p<a_j \\
      \left(0,(\alpha'-\alpha^{ij})\left(\sqrt{\frac{b_j}{1-b_j}}-\sqrt{\frac{a_j}{1-a_j}}\right)\right) & \text{ if } p>b_j \\
      \left((\alpha'-\alpha^{ij})\left(\sqrt{\frac{1-p}{p}}-\sqrt{\frac{1-b_j}{b_j}}\right),(\alpha'-\alpha^{ij})\left(\sqrt{\frac{p}{1-p}}-\sqrt{\frac{a_j}{1-a_j}}\right)\right) & \text{ if } p\in[a_j,b_j]
    \end{cases}$
    \State $(\res,\alpha^{ij}) \gets (\res+\res',\alpha')$ 
    \EndFunction
    \medskip
    \Function{ExecuteTrade}{$\r \in \reals^2$}
    \State Let $p=\textproc{price}(\res),\quad \hat{p}=\textproc{price}(\res-\r)$ \footnotemark
    \State Let $l,u$ be such that $a_l\leq p \leq b_l$ and $a_u\leq \hat{p}\leq b_u$.
    
    \State \textbf{check} 
    \tiny{\begin{equation*}
        \frac{1}{(\sum_{i=0}^k\alpha^{il})^2}\left(\xa+\sum_{i=0}^k\alpha^{il}\sqrt{\frac{1-b_l}{b_l}}\right)\left(\xb+\sqrt{\frac{a_l}{1-a_l}}\sum_{i=0}^k\alpha^{il}\right)=\frac{1}{(\sum_{i=0}^k\alpha^{iu})^2}\left(\xa-r_1+\sum_{i=0}^k\alpha^{iu}\sqrt{\frac{1-b_u}{b_u}}\right)\left(\xb-r_2+\sqrt{\frac{a_u}{1-a_u}}\sum_{i=0}^k\alpha^{iu}\right)
      \end{equation*}}
    \normalsize{\label{item:V3-trade-check}
      \State \textbf{pay} $\beta\frac{\sum_j\alpha^{ij}}{\sum_j\sum_{o}\alpha^{oj}}(-\r)_+$ to each LP $i$ where $j$ sums over buckets in $[B^l,B^u]$.\Comment{WLOG assume that the $B^u$ bucket comes later than $B^l$}
      \label{item:V3-fee-split}
      \State $\res\gets\res-\r $
      \EndFunction}
  \end{algorithmic}
\end{algorithm}
Readers familiar with the original protocol may recognize Protocol \ref{alg:V3-two-asset} as Uniswap V3 mechanics but with minor changes coming from using normalized prices.
Line \ref{item:V3-liquidity-request} comes from \cite{fan2022differentialliquidity}'s analysis of Uniswap V3, and Line \ref{item:V3-trade-check} comes from the shifted reserve curve characteristic of Uniswap V3 as seen in both \cite{fan2022differentialliquidity} and \cite{adams2021whitepaperv3}.
For clarity, we want to reiterate that \citet{fan2022differentialliquidity} uses an exchange rate price $\hat{p}$, and we use its normalized version $p$. The two quantities are related by $\hat{p}=\frac{p}{1-p}$.
\begin{proposition}\label{prop:is_uniswap_V3}
  For $\Ginit = \{\sum_j\alpha_j g^{(j)} \mid \alpha_j > 0\}$, $\GLP = \{\sum_j\alpha_j g^{(j)} \mid \alpha_j \geq 0\}$ where 
  \begin{align*}
    g^{(j)}(p)=\begin{cases}
      p(\sqrt{\frac{1-b_j}{b_j}}-\sqrt{\frac{1-a_j}{a_j}}) & \text{ if } p\leq a_j \\
      -2\sqrt{p(1-p)}+p\sqrt{\frac{1-b_j}{b_j}} +(1-p)\sqrt{\frac{a_j}{1-a_j}} & \text{ if } a_j \leq p \leq b_j\\
      (1-p)(\sqrt{\frac{a_j}{1-a_j}}-\sqrt{\frac{b_j}{1-b_j}}) &\text{ if } p \geq b_j
    \end{cases},
  \end{align*}
  the vector $\q$ of liability in Protocol \ref{alg:general-two-asset-g} always satisfies $\left(\xa+\alpha_j\sqrt{\frac{1-b_j}{b_j}}\right)\cdot \left(\xb+\alpha_j\sqrt{\frac{a_j}{1-a_j}}\right) = \alpha_j^2$ for some $\alpha_j > 0$, where $\res=-\q$.
\end{proposition}
\begin{proof}
  Observe that any change in liability vector in the \textsf{ModifyLiquidity} or \textsf{Initialize} phases results in a liability vector that results in the vector taking the form $\liabilityf(g,p)$ for some $g\in \GLP$, $p\in[0,1]$. 
  Let this $g$ be $\alpha_jg^{(j)}(p)$ where $\alpha_j$ is the total liquidity in $j$th price interval which is $\sum_i\alpha^{ij}$ in Protocol \ref{alg:V3-two-asset}. We have
  \begin{align*}
    g'(p)=\begin{cases}
      \alpha_j(\sqrt{\frac{1-b_j}{b_j}}-\sqrt{\frac{1-a_j}{a_j}}) & \text{ if } p\leq a_j \\
      \alpha_j(-\frac{1-2p}{\sqrt{p(1-p)}}+\sqrt{\frac{1-b_j}{b_j}} -\sqrt{\frac{a_j}{1-a_j}}) & \text{ if } a_j \leq p \leq b_j\\
      (\sqrt{\frac{b_j}{1-b_j}}-\sqrt{\frac{a_j}{1-a_j}}) &\text{ if } p \geq b_j
    \end{cases}.
  \end{align*}
  Solving for $\liabilityf(g,p)=(g(p)-p\cdot g'(p)\ones + \begin{pmatrix}
    g'(p) \\ 0
  \end{pmatrix}$ for these three cases gives us
  \begin{align*}
    \liabilityf(g,p)&= \begin{cases}
      \alpha_j\begin{pmatrix}
        \sqrt{\frac{1-b_j}{b_j}}-\sqrt{\frac{1-a_j}{a_j}}\\0
      \end{pmatrix}& \text{ if } p\leq a_j \\
      \alpha_j\begin{pmatrix}
        -\sqrt{\frac{1-p}{p}}+\sqrt{\frac{1-b_j}{b_j}}\\-\sqrt{\frac{p}{1-p}}+\sqrt{\frac{a_j}{1-a_j}}
      \end{pmatrix}& \text{ if } a_j \leq p \leq b_j\\
      \alpha_j\begin{pmatrix}
        0\\\sqrt{\frac{a_j}{1-a_j}}-\sqrt{\frac{b_j}{1-b_j}})
      \end{pmatrix} &\text{ if } p \geq b_j
    \end{cases}.
  \end{align*}
  \footnotetext{We also note that price in Uniswap V3 is not calculated on demand like we do here but is a state variable thats maintained throughout the implementation.
    Uniswap V3 also implements the Line \ref{item:V3-trade-check} by passing through all price ranges consecutive to $p$ and checking which price interval satisfies this check.
    We abstract away from this to avoid being caught up in technicalities as this is not the main problem we tackle.}
  In each of these cases, with a bit of algebra we can see that 
  \begin{equation*}
    \left(q_1-\alpha_j\sqrt{\frac{1-b_j}{b_j}}\right)\cdot \left(q_2-\alpha_j\sqrt{\frac{a_j}{1-a_j}}\right)= \left(\xa+\alpha_j\sqrt{\frac{1-b_j}{b_j}}\right)\cdot \left(\xb+\alpha_j\sqrt{\frac{a_j}{1-a_j}}\right) = \alpha_j^2,
  \end{equation*}
  as $\q$ in Protocol \ref{alg:general-two-asset-g} is liability which is negative of reserves in Protocol \ref{alg:V3-two-asset}.
\end{proof}

\begin{proposition}
  Protocol \ref{alg:V3-two-asset} is equivalent to  Protocol \ref{alg:general-two-asset-g} for $\Ginit = \{\sum_j\alpha_j g^{(j)} \mid \forall j\, \alpha_j > 0\}$, $\GLP = \{\sum_j\alpha_j g^{(j)} \mid \forall j\, \alpha_j \geq 0\}$ where 
  \begin{align*}
    g^{(j)}(p)=\begin{cases}
      p(\sqrt{\frac{1-b_j}{b_j}}-\sqrt{\frac{1-a_j}{a_j}}) & \text{ if } p\leq a_j \\
      -2\sqrt{p(1-p)}+p\sqrt{\frac{1-b_j}{b_j}} +(1-p)\sqrt{\frac{a_j}{1-a_j}} & \text{ if } a_j \leq p \leq b_j\\
      (1-p)(\sqrt{\frac{a_j}{1-a_j}}-\sqrt{\frac{b_j}{1-b_j}}) &\text{ if } p \geq b_j
    \end{cases}~.
  \end{align*}
\end{proposition}
\begin{proof}
  To show that Protocol \ref{alg:general-two-asset-g} gives us Protocol \ref{alg:V3-two-asset}, for the choices of $g$ specified, involves showing three specific components are equivalent.
  They include showing that the initialization check satisfies and the request vectors in \textsf{ModifyLiquidity} and \textsf{ExecuteTrade} routines match.

  We saw $\alpha_j$ to mean total liquidity in $j$th interval i.e. $\sum_{i=0}^k\alpha^{ij}$ and let $g(p)=\alpha_j g^{(j)}$.
  Firstly, to show the initialization check, we derive that 
  \begin{align*}
    (g')^{-1}(q_1-q_2)=(g')^{-1}(\xb-\xa)=\frac{\xb+\alpha_j\sqrt{\frac{a_j}{1-a_j}}}{\xa+\alpha_j\sqrt{\frac{1-b_j}{b_j}}+\xb+\alpha_j\sqrt{\frac{a_j}{1-a_j}}}.
  \end{align*}
  We consider only the case when this price falls in a bucket $j$ as for all other buckets, $(g')^{-1}$ would not give a definitive result.
  \begin{align*}
    \liabilityf(g_0,\pricef(g_0,q^0))&=\alpha_j\begin{pmatrix}
      -\sqrt{\frac{\xa^0+\alpha_j\sqrt{\frac{1-b_j}{b_j}}}{\xb^0+\alpha_j\sqrt{\frac{a_j}{1-a_j}}}}+\sqrt{\frac{1-b_j}{b_j}}\\-\sqrt{\frac{\xb^0+\alpha_j\sqrt{\frac{a_j}{1-a_j}}}{\xa^0+\alpha_j\sqrt{\frac{1-b_j}{b_j}}}}+\sqrt{\frac{a_j}{1-a_j}}
    \end{pmatrix}\\
                                     &=\alpha_j\begin{pmatrix}-\frac{\xa^0+\alpha_j\sqrt{\frac{1-b_j}{b_j}}}{\alpha_j}+\sqrt{\frac{1-b_j}{b_j}}\\-\frac{\xb^0+\alpha_j\sqrt{\frac{a_j}{1-a_j}}}{\alpha_j}+\sqrt{\frac{a_j}{1-a_j}}
                                     \end{pmatrix}\\
                                     &= \begin{pmatrix}
                                       -\xa^0 \\ -\xb^0
                                     \end{pmatrix}  = \begin{pmatrix}
                                       q_1^0 \\ q_2^0
                                     \end{pmatrix}
  \end{align*}
  The first two steps are a result of Proposition \ref{prop:is_uniswap_V3}.
  
  Now, we show that the quantity of assets requested from a LP to change the liquidity level from $\alpha^{ij}$ to $\alpha'$ given by Line \ref{item:V3-liquidity-request} in Protocol \ref{alg:V3-two-asset} is equivalent to Line \ref{item:request_q_general_g} of Protocol \ref{alg:general-two-asset-g}.

  We first note that the price $p$, does not change in this operation, as liquidity must be added by keeping the price constant.
  Another way to see it is that $\pricef(g_0,\q^0)$ remains unchanged.
  A liquidity change of $\alpha^{ij}$ to $\alpha'$ at price $p\in B_j$ reflects a change of $g_i$ from $\alpha^{ij}g^{(j)}$ to $\hat{g}=\alpha'g^{(j)}$ in Uniswap V3.
  The quantity of asset requested by Protocol \ref{alg:general-two-asset-g} is given by
  \begin{align*}
    -(\liabilityf(\hat g,p) - \q^i)&=\liabilityf(g_i,p) - \liabilityf(\hat{g},p) \\
                                   &=  \begin{cases}
                                     \alpha^{ij}\left(-\sqrt{\frac{1-a_j}{a_j}}+\sqrt{\frac{1-b_j}{b_j}},0\right)-\alpha'\left(-\sqrt{\frac{1-a_j}{a_j}}+\sqrt{\frac{1-b_j}{b_j}},0\right) & \text{ if } p<a_j \\
                                     \alpha^{ij}\left(0,\sqrt{\frac{a_j}{1-a_j}}+\sqrt{\frac{b_j}{1-b_j}}\right)-\alpha'\left(0,\sqrt{\frac{a_j}{1-a_j}}+\sqrt{\frac{b_j}{1-b_j}}\right) & \text{ if } p>b_j \\
                                     (\alpha^{ij}-\alpha')\left(-\sqrt{\frac{1-p}{p}}+\sqrt{\frac{1-b_j}{b_j}},\sqrt{\frac{p}{1-p}}+\sqrt{\frac{a_j}{1-a_j}}\right) & \text{ if } p\in[a_j,b_j]
                                   \end{cases}\\
                                   &= \begin{cases}
                                     \left((\alpha'-\alpha^{ij})\left(\sqrt{\frac{1-a_j}{a_j}}-\sqrt{\frac{1-b_j}{b_j}}\right),0\right) & \text{ if } p<a_j \\
                                     \left(0,(\alpha'-\alpha^{ij})\left(\sqrt{\frac{b_j}{1-b_j}}-\sqrt{\frac{a_j}{1-a_j}}\right)\right) & \text{ if } p>b_j \\
                                     \left((\alpha'-\alpha^{ij})\left(\sqrt{\frac{1-p}{p}}-\sqrt{\frac{1-b_j}{b_j}}\right),(\alpha'-\alpha^{ij})\left(\sqrt{\frac{p}{1-p}}-\sqrt{\frac{a_j}{1-a_j}}\right)\right) & \text{ if } p\in[a_j,b_j]
                                   \end{cases}.
  \end{align*}

  Now, we show that the check in Line \ref{item:general_aggregate_check} of Protocol \ref{alg:general-two-asset-g} for the given $g$ gives us the condition $\left(\xa+\alpha_{l}\sqrt{\frac{1-b_l}{b_l}}\right)\left(\xb+\sqrt{\frac{a_l}{1-a_l}}\alpha_l\right)=\left(\xa-r_1+\alpha_u\sqrt{\frac{1-b_u}{b_u}}\right)\left(\xb-r_2+\sqrt{\frac{a_u}{1-a_u}}\alpha_u\right)$ where $\alpha_x = \sum_{i=0}^k \alpha^{ix}$ that appears in Uniswap V3.

  Let $\q=\sum_i\q^i$ and $\r=\sum_i\r^i$.
  The check in Protocol \ref{alg:general-two-asset-g} can be rewritten as 
  \begin{align*}
    \q+\r&= \liabilityf\left(\sum_{i=0}^k g_i,\hat{p}\right)\\
         &= \liabilityf\left(\sum_{i=0}^k \alpha^{iu}g^{(j)},\hat{p}\right)\\
         &= \sum_{i=0}^k \alpha^{iu}\liabilityf(g^{(j)},\hat{p}).
  \end{align*}
  From Proposition \ref{prop:is_uniswap_V3}, we can see that
  \begin{equation*}
    \left(\xa-r_1+\left(\sum_{i=0}^k \alpha^{iu}\right)\sqrt{\frac{1-b_u}{b_u}}\right)\cdot \left(\xb-r_2+\left(\sum_{i=0}^k \alpha^{iu}\right)\sqrt{\frac{a_u}{1-a_u}}\right) = (\sum_{i=0}^k \alpha^{iu})^2
  \end{equation*}
  and 
  \begin{equation*}
    \left(\xa+\sum_{i=0}^k \alpha^{il}\sqrt{\frac{1-b_l}{b_l}}\right)\cdot \left(\xb+\sum_{i=0}^k \alpha^{il}\sqrt{\frac{a_l}{1-a_l}}\right) = \left(\sum_{i=0}^k \alpha^{il}\right)^2,
  \end{equation*}
  proving what we need.

  The last fact we want to show to complete the proof is that $r^i=\frac{\alpha^i}{\alpha} r$ for Uniswap V3, where $\alpha_i=\sum_j\alpha^{ij}$, $\alpha=\sum_j\sum_{o=0}^k\alpha^{oj}$ for $j$ summing over baskets $B^l$ to $B^u$.
  We see that 
  \begin{align*}
    \r^i &=\liabilityf(g_i,\hat{p})-\q^i\\
         &= \alpha^i(\liabilityf(g^{(j)},\hat{p})-\liabilityf(g^{(j)},p))\\
         &= \frac{\alpha^i}{\alpha}\left(\liabilityf(g,\hat{p}) -\q \right) \text{\quad where $g=\sum_{i=0}^k g_i$ and $\q=\sum_{i=0}^k\q^i$}\\
         &= \frac{\alpha^i}{\alpha}\r,
  \end{align*}
  as desired.
\end{proof}
\section{On new protocols} \label{sec:new_protocols_detailed}

In some cases, the generalized bucketing scheme we provide in \S~\ref{sec:general_bucketing} is still overly restrictive on the expressivity of LPs, depending on the size of the price intervals.
It may also be unnatural for LPs to specify their liquidity allocation via discontinuous liquidity functions.
Fortunately, our general protocol easily allows one to generate more expressive restricted protocols, for particular families of $C$ (equivalently $g$ for two assets) functions which still allow for efficient computations.
For example, one could use ``soft'' liquidity buckets where liquidity continuously ``fades'' in and out around a target price $a_j$.

In the case of $n$ assets, one can generalize the idea of buckets, though now the task of partitioning the $d=n-1$-dimensional price space becomes nontrivial.
One promising approach would be to partition into a cell complex of convex polytopes, specifically a power diagram, \citep{aurenhammer1987powerd}.
One could define indicator functions as above to allow LPs to allocate liquidity uniformly (or according to some base shape) across each of these regions.
We also give a piecewise linear protocol which could be similarly adapted to these polyhedral regions, where the linear pieces align with the regions.

\subsection{Soft buckets allowing richer functions than discrete buckets}

Define a set $a_0 < a_1 = 0 < a_2 < \cdots < a_k = 1 < a_{k+1}$.
Our ``buckets'' $B_j$ will be supported on the interval $[a_{j-1},a_{j+1}]$.
To capture the continuous fading, define the triangular function $T^{(j)}:[0,1]\to[0,1]$ as $T^{(j)}(p) = \left(\frac{p-a_{j-1}}{a_j-a_{j-1}}\right)\ones_{p\in[a_{j-1},a_j]}+\left(\frac{a_{j+1}-p}{a_{j+1}-a_{j}}\right)\ones_{p\in[a_{j},a_{j+1}]}$.
The corresponding base liquidity functions are then $\ell^{(j)} = (\ell T^{(j)})(p)$ 
where $\ell(p)=2(p(1-p))^{-\frac{3}{2}}$, to again use the same base ``shape'' as the constant product invariant $\varphi_\alpha(\res)=\xa\xb=\alpha^2$.
Let $g^{(j)}$ then be the corresponding base generating function for $\ell^{(j)}$.

\begin{figure}
 \begin{center}
   \begin{tabular}{cc}
     \includegraphics{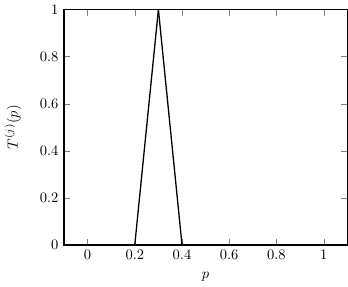}
     &
     \includegraphics{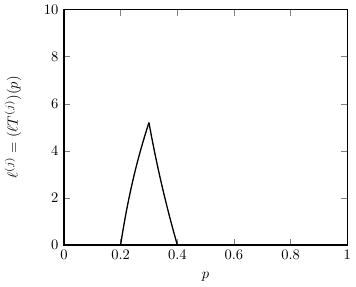}
   \end{tabular}
   \caption{$T^{(j)}(p)$ and $\ell^{(j)}(p)$ for $a_j=0.3$}
 \end{center}
\end{figure}

Now letting
$\Ginit = \{\sum_j\alpha_j g^{(j)} \mid \forall j\, \alpha_j > 0\}$
and $\GLP = \{\sum_j\alpha_j g^{(j)} \mid \forall j\, \alpha_j \geq 0\}$, Protocol~\ref{alg:general-two-asset-g} gives a new liquidity provisioning protocol.
While the corresponding computations appear to be essentially as light-weight as Uniswap V3, the liquidity function is better behaved. 
We can characterize the possible liquidity functions $\{g'' \mid g \in \GLP\}$ as the functions $f \ell$ where $f$ is an arbitrary nonnegative continuous function that is affine on each interval $[a_j,a_{j+1}]$.
(Simply take $\alpha_j = f(a_j)$.)
Thus, we have in particular that liquidity is always continuous in the price in this new protocol.
One can additionally innovate from here, adding more flexibility, while taking care to keep the various computations manageable.
\subsection{Piecewise linear market maker}

To demonstrate the robustness of our protocol, in this section we consider a market where the liquidity curves $\ell$ are not well-defined.
Yet, we show that our general scheme still gives rise to a sensible market maker. 

Consider a market maker that restricts the possible prices to $\{p^i\}_{i=1}^m$ where $0 < a_1 < \cdots < a_m < 1$ .
Consider a $(g^{(j)})'$  of the form 
\begin{equation*}
 (g^{(j)})'(p)=\begin{cases}
   a_j-1 & p \leq a_j \\
   [a_j-1,a_j]& p=a_j\\
   a_j & p\geq a_j
 \end{cases}.
\end{equation*}
$g'(p)=\sum_j \alpha_j (g^{(j)})'(p)$ where $\alpha_j = \sum_i\alpha^{ij}$.

Observe that liquidity is infinity at each price.
As prices only move discretely in this market, the market needs to be parameterized by reserves held / liability vector $\q$.
Hence the state of the market is $\{\alpha_{ij}\}_{i\in\{0,\cdots,k\},j\in\{1,\cdots ,m\}},\{\q^i\}_{i\in\{0,\cdots, k\}}$.

If the current market reserves are $\q$, the price in this market can be derived from the below formula 

\begin{align*}
   \pricef(g,\q) &= a_{j^*} \text{ where } j^*=\argmax_{j'} \text{ s.t. } \left\{\q-\sum_{j=1}^m\alpha_j(a_j-1)+\sum_{j=1}^{j'-1} \alpha_j \geq 0\right\}.
\end{align*}

For a given $\q$, let there exist $y\in[0,1)$ such that $\q=\sum_j \alpha_ja_j - \sum_{j=j^*}^m \alpha_j+y\alpha_{j^*}$.
We can maintain this relative liquidity in a bucket after changing the liquidity curves.

For the \textsf{ModifyLiquidity} function, let LP $i$ wants to change its liquidity levels from $\alpha_{ij}$ to $\alpha_{ij}'$ for price bucket $j$.
The new reserves that LP $i$ needs to deposit is given by $\hat{\q}-\q =(\alpha'_{ij}-\alpha_{ij})(a_j-\ones_{j\geq j^*}+y\ones_{j=j^*})$

\textsf{ExecuteTrade} takes in a trade $r$, computes the new price $\hat{p}$ corresponding to the new liability vector $\q+\r$ using the above given formula.

\section{On incomplete markets} \label{sec:app-incomplete-markets}

In the main body, we stated our general framework in terms of complete markets. Here, we describe how Protocol~\ref{alg:general-n-asset-cost} works for incomplete markets.
We follow \citet{abernethy2013efficient} to define an incomplete market with $k$ securities by specifying a function $\phi:\{1,\ldots,n\}\to\mathbb{R}^k$ where the $i$th security is $\phi_i$.
In \citet{abernethy2013efficient}, cost functions are given by the convex conjugate $\hat{C}:\mathbb{R}^k\to\mathbb{R}$ of some $F:\Pi\to\mathbb{R}$ where $\Pi=\conv(\{\phi(i):i\in\{1,\ldots,n\}\})$ is the price space (where $F$ is denoted by $R$ in their paper).
Alternatively, we can have securities live in $\mathbb{R}^n$, and \emph{extend} $F$ to $G:\Delta_n\to\mathbb{R}$ with
$$
G(\p)=F(E_{\p}\phi),\quad E_{\p}\phi=\sum_{i=1}^n p_i\phi_i,
$$
where $E_{\p}$ is the expected payoff of the securities.
Then, we can take the convex conjugate to get a cost function $C$ for Protocol~\ref{alg:general-n-asset-cost}.
We can observe that $G$ is flat along directions in which the expected payoff is constant, so $C$ has zero liquidity along these dual directions.
Hence, the share vectors $\q$ are effectively constrained to $\Span(\{\phi_1,\ldots,\phi_k\})$.
In this sense, incomplete markets are a special case of Protocol~\ref{alg:general-n-asset-cost}.

It is worth noting that it is easier to work in the lower-dimensional space $\mathbb{R}^k$ directly, as \citet{abernethy2013efficient} also do.
For our purposes, the protocols and equivalence results are easier to state, and \citet{abernethy2013efficient} provide similar equivalences but with stronger assumptions on $F$ and $\hat{C}$, namely strict convexity of $F$.
Since we do not require $G$ to be strictly convex in Theorem~\ref{thm:equivalence-interpretations}, our equivalence theorem applies to the incomplete market setting. 
\end{document}